\documentclass{lmcs}
\pdfoutput=1

\usepackage{lastpage}
\renewcommand{\dOi}{14(4:18)2018}
\lmcsheading{}{1--\pageref{LastPage}}{}{}%
{Mar.~07,~2017}{Nov.~26,~2018}{}

\keywords{proof theory, cut-elimination, labelled sequents, nested sequents, modal logic, intermediate logics}

\usepackage{bussproofs}
\usepackage{amsmath}
\usepackage{amssymb}
\usepackage[all]{xy}

\newcommand{\REL}[1][]{\mathcal{R}#1,\mathcal{E}}

\newcommand{\DC}{\vdash}

\newcommand{\LTSE}{\text{LTSE}}

\newcommand{\labKst}{\text{\textbf{LSEqK*}}}
\newcommand{\INSK}{\text{\textbf{INSK}}}
\newcommand{\INSKst}{\text{\textbf{INSK*}}}

\newcommand{\labI}{\text{\textbf{LSEqIp}}}
\newcommand{\labK}{\text{\textbf{LSEqK}}}
\newcommand{\imp}{\rightarrow}

\newcommand{\INSI}{\text{\textbf{INSIp}}}

\newcommand{\LTS}{\text{LTS}}

\renewcommand{\Box}{\square}
\newcommand{\diam}{\Diamond}

\newcommand{\context}{}
\newcommand{\contexta}{}
\newcommand{\contextb}{}

\newcommand{\bl}{}

\newcommand{\ob}{\{}
\newcommand{\cb}{\}}

\newcommand{\RR}{\mathsf{R}}
\newcommand{\RRc}{\hat{\mathsf{R}}}

\newcommand{\sub}[2]{\{#1/#2\}}

\newcommand{\g}{g}
\renewcommand{\L}{L}
\newcommand{\I}{I}

\begin{document}

\title[Inducing syntactic cut-elimination for indexed nested sequents]{Inducing syntactic cut-elimination for indexed nested sequents\rsuper*}
\titlecomment{{\lsuper*}This is an extended version of the conference paper~\cite{Ramijcar16}.\\
Work supported by the FWF projects: START Y544-N23 and I 2982.}

\author{Revantha Ramanayake}
\address{Technische Universit\"{a}t Wien, Vienna, Austria.}
\email{revantha@logic.at}

\begin{abstract}
The key to the proof-theoretic study of a logic is a proof calculus with a subformula property. 
Many different proof formalisms have been introduced (e.g. sequent, nested sequent, labelled sequent formalisms) in order to provide such calculi for the many logics of interest. The nested sequent formalism was recently generalised to indexed nested sequents in order to yield proof calculi with the subformula property for extensions of the modal logic~K by (Lemmon-Scott) Geach axioms. The proofs of completeness and cut-elimination therein were semantic and intricate. Here we show that derivations in the labelled sequent formalism whose sequents are `almost treelike' correspond exactly to indexed nested sequents. This correspondence is exploited to induce syntactic proofs for indexed nested sequent calculi making use of the elegant proofs that exist for the labelled sequent calculi. A larger goal of this work is to demonstrate how specialising existing proof-theoretic transformations alleviate the need for independent proofs in each formalism. Such coercion can also be used to induce new cutfree calculi. We employ this to present the first indexed nested sequent calculi for intermediate logics. 
\end{abstract}

\maketitle

\section{Introduction}\label{sec_intro}

A standard syntactic method of presenting a logic is via its Hilbert (proof) calculus. The Hilbert proof calculus for intuitionistic logic typically consists of a finite set of axiom schemata and a single rule called \textit{modus ponens}.
\[
  \AxiomC{$A$}
  \AxiomC{$A\imp B$}
  \RightLabel{\textit{modus ponens}}
  \BinaryInfC{$B$}
  \DisplayProof
\]
A proof (\textit{derivation}) is a finite sequence of formulae such that every element is either an instance of an axiom or obtainable from previous elements via \textit{modus ponens}.
A formula is called a \textit{theorem} of the logic if it has a derivation. This presentation of the logic is simple and elegant. It is also modular: a Hilbert calculus for classical logic can be obtained from the Hilbert calculus for intuitionistic logic by the addition of a single axiom schema $p\lor\lnot p$ corresponding to the \textit{excluded middle}.

However, the rule of \textit{modus ponens} has a great drawback. To investigate if a formula~$B$ that is not an instance of an axiom is a theorem, it is neccesary to investigate if there is some~$A$ and~$A\imp B$ that is a derivable (since~$B$ must be the conclusion of a \textit{modus ponens} rule instance), and there is no indication of what~$A$ we should choose for this purpose. Indeed, in terms of formula size, $A$~may be much larger than~$B$.  In other words, even if~$B$ \textit{is} a theorem, the formula appears to contain little information on the form of its derivation.


Gentzen~\cite{Gen69} showed that there is an alternative, presenting proof calculi for classical and intuitionistic logic in a new formalism which he called the \textit{sequent calculus}. Unlike the afore-mentioned Hilbert calculi, these sequent calculi satisfied the \textit{subformula property}: every formula in the premise of a rule is a (sub)formula of some formula in the conclusion of the rule. When every rule in the proof system satisfies the subformula property, then the only formulae that may appear in a derivation are the subformulae of the formula that is being proved. Thus the subformula property places a strong restriction on the set of possible derivations of a given formula. 
Apart from facilitating backward proof search, the subformula property (and the ensuing restriction on the derivations) enable elegant and constructive proof-theoretic arguments of logical properties such as consistency, decidability, complexity and interpolation. We note in passing that weaker versions of the subformula property---e.g. that every formula in the derivation belonging to a finite set computable from the formula that is being proved---would also suffice.

The building block of the sequent calculus is not a formula like in the Hilbert calculus, but instead a tuple of formula multisets (formula lists in the original formulation), written $X\DC Y$ (we shall call this a \textit{traditional sequent}). The sequent calculus consists solely of rules; every rule consists of some number of premises and a conclusion, and each is a traditional sequents. The starting point for a derivation is now a rule with no premises (an \textit{initial sequent}). In particular, Gentzen traded the axiom schemata of the Hilbert calculus for rule schemata which either add a single logical connective in the conclusion in positive/negative position, or do not introduce introduce any new logical connectives. The point is that these rules satisfy the subformula property. With these new rules and sequent-version of \textit{modus ponens}, every theorem could be derived. However since this sequent calculus contains \textit{modus ponens} it still violates the subformula property. To establish that eliminating \textit{modus ponens} does not limit the expressivity of the calculus, Gentzen proved a stronger result by first replacing \textit{modus ponens} with the \textit{cut rule} below and then proving his celebrated \textit{cut-elimination theorem} which shows how to eliminate the cut rule from any derivation, thus showing that the cut rule is redundant. 
\[
  \AxiomC{$X\DC Y,A$}
  \AxiomC{$A,U\DC V$}
  \RightLabel{cut}
  \BinaryInfC{$X,U\DC Y,V$}
  \DisplayProof
\]
Gentzen used the cut-elimination result to give a proof of the consistency of arithmetic using a suitable induction principle. 

Following the seminal work of Gentzen---and recognising its proof-theoretic importance---efforts were made to construct sequent calculi with the subformula property for the many non-classical logics of interest. Unfortunately, in most cases, the restrictions that the sequent calculus formalism enforces (most prominently, the form of the traditional sequent as a formula multiset/list tuple) does not provide enough leeway to construct rules to capture these logics. In response, various extensions of the sequent calculus formalism have been proposed (some formalisms not discussed in this paper include hypersequents~\cite{Pot83,Avr96} and the display calculus~\cite{Bel82}). The typical approach is to extend/generalise the traditional sequent data structure (an exception is the calculus of structures~\cite{Gug07}).

%
The \textit{nested sequent}~\cite{Kas94,Bru06} formalism has been used to give proof calculi for intuitionistic logic~\cite{Fit12}, conditional logics~\cite{AleOliPoz16}, logics in the classical and intuitionistic modal cube~\cite{Kas94,Str13FOSSACS,MarStr14AIML} and path axiom extensions of classical modal logic~\cite{GorPosTiu11}.
A nested sequent replaces the traditional sequent data structure with a richer data structure.
\begin{defi}[nested sequent]
A \emph{nested sequent} (NS) is a finite object defined recursively as follows:
\begin{align*}
&\text{NS} := X\DC Y\text{ where $X$ and~$Y$ are formula multisets}	\\
&\text{NS} := \text{NS}, [\text{NS}], \ldots, [\text{NS}]
\end{align*}
\end{defi}
\noindent Here are some examples of nested sequents.
\begin{align}
&X\DC Y, [P\DC Q, [U\DC V], [L\DC M]], [S\DC T] \label{NSone} \\
&X\DC Y , [U\DC V , [L\DC M],[S\DC T]]  \label{NStwo}
\end{align}
A nested sequent lends itself naturally to a graphical representation as a directed tree where each node is decorated with a traditional sequent.
Here are the representations of the preceding two nested sequents:
\medskip
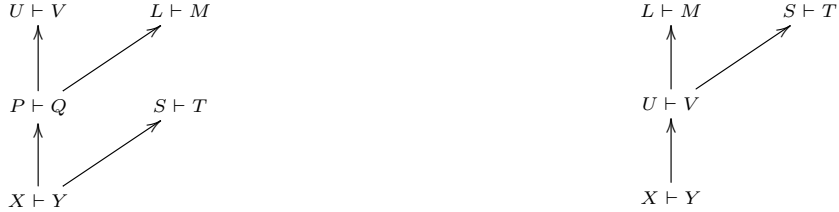
\begin{center}
\begin{tiny}
\begin{tabular*}{0.75\textwidth}{@{\extracolsep{\fill}}cc}
\xymatrix{
U\DC V& L\DC M \\
P\DC Q\ar[u]\ar[ur] & S\DC T\\
X\DC Y\ar[u]\ar[ur] & 
}
&
\xymatrix{
L\DC M&S\DC T\\
U\DC V\ar[u]\ar[ur] &\\
X\DC Y\ar[u] &
}
\end{tabular*}
\end{tiny}
\end{center}
\medskip
Unfortunately, even the nested sequent formalism is insufficient to provide calculi with the subformula property for many logics of interest.
In the context of modal logic, the axiomatic extension of the normal modal logic~$K$ by \textit{Geach axioms} is such a class of logics.
One reason for the interest is that these axioms capture many interesting reasoning properties (e.g. reflexivity, symmetry, transitivity, confluence). Another reason is that these logics possess an elegant semantics: Kripke frames~\cite{BlaRijVen01} which satisfy the corresponding first-order property.
\begin{defi}[Geach logics and axioms]
A \emph{Geach} logic is an axiomatic extension of the normal modal logic~$K$ by finitely many (Geach) axioms of the form:
\begin{equation}\label{Geach_axioms}
G(h,i,j,k):= \diam^{h}\Box^{i} p\imp \Box^{j}\diam^{k} p\quad (h,i,j,k\geq 0)
\end{equation}
\end{defi}
For some Geach logics (e.g. $K+\diam\Box p\imp\Box\diam p$) even the nested sequent formalism (ultimately, the nested sequent data structure) appears not expressive enough to provide a calculus with the subformula property.

Recently, Fitting~\cite{Fit15} extended the nested sequent data structure by assigning an index to each traditional sequent (permitting multiple traditional sequents to possess the same index) and introduced the \textit{indexed nested sequent} formalism to present proof calculi with the subformula property for those Geach logics whose Geach axioms satisfy~$i+k>0$.
\begin{defi}[indexed nested sequent]
An \emph{indexed nested sequent}  (INS) is a finite object defined recursively as follows:\par
\begin{align*}
&\text{INS} := X\DC^{a} Y\text{ where $X$ and~$Y$ are formula multisets and $a\in\mathbb{N}$}	\\
&\text{INS} := \text{INS}, [\text{INS}], \ldots, [\text{INS}]
\end{align*}
\end{defi}
Here are some examples of indexed nested sequents:
\begin{align}
&X\DC^{0} Y, [P\DC^{1} Q, [U\DC^{2} V], [L\DC^{3} M]], [S\DC^{3}  T]\label{INSone}	\\
&X\DC^{0} Y, [P\DC^{1} Q, [U\DC^{0} V], [L\DC^{1} M]], [S\DC^{2} T]\label{INStwo}
\end{align}
An indexed nested sequent can be represented graphically as a directed tree whose nodes are decorated by traditional sequents and an index (the first two graphs below represent~(\ref{INSone}) and~(\ref{INStwo}) respectively.
\begin{center}
  \vspace{-0.75em}
\begin{tiny}
\begin{tabular}{cc@{\hspace{1em}}|@{\hspace{1em}}cc}
\xymatrix{
U\DC^{2} V& L\DC^{3} M \\
P\DC^{1} Q\ar[u]\ar[ur] & S\DC^{3} T\\
X\DC^{0} Y\ar[u]\ar[ur] & 
}
&
\xymatrix{
U\DC^{0} V& L\DC^{1} M \\
P\DC^{1} Q\ar[u]\ar[ur] & S\DC^{2} T\\
X\DC^{0} Y\ar[u]\ar[ur] & 
}
&
\xymatrix{
U\DC^{2} V& \\
P\DC^{1} Q\ar[u]\ar[r] & S,L\DC^{3} T,M\\
X\DC^{0} Y\ar[u]\ar[ur] & 
}
&
\xymatrix{
& \\
P,L\DC^{1} Q,M\ar@(dl,ul)[d]\ar@(ur,ul) & S\DC^{2} T\\
X,U\DC^{0} Y,V\ar[u]\ar[ur] & 
}
\end{tabular}
\end{tiny}
\end{center}
\medskip
The key to obtaining more expressivity is to interpret those nodes of the tree with the same index ($X_{1}\DC^{a} Y_{1},\ldots,X_{N+1}\DC^{a}Y_{N+1}$, say) as a \textit{single} node decorated with $X_{1},\ldots,X_{N+1}\DC Y_{1},\ldots,Y_{N+1}$.
Thus INS correspond to \textit{directed graphs} obtained by conflating certain nodes of a directed tree. The INS~(\ref{INSone}) and~(\ref{INStwo}) are depicted by the two rightmost graphs above.

Fitting does not prove \textit{syntactic cut-elimination} for his indexed nested sequent calculi for Geach logics, but instead establishes that the calculus minus the cut rule is complete with respect to the corresponding logic's semantics. Such a proof is called a \textit{semantic proof of cut-elimination} and contrasts with the syntactic proofs \`{a} la Gentzen which rely solely on the proof calculus.

Apart from the technical interest in syntactic proofs of cut-elimination (after all, proof-theory is concerned primarily with the syntax), syntactic proofs yield a constructive procedure so the cutfree derivation is related in a formal sense to the original derivation. However, the downside of syntactic proofs is that they tend to be highly technical and difficult to verify. 
We believe that the best response is to reuse and whenever possible adapt those syntactic proofs that are already in existence rather than constructing new proofs from scratch. In Ramanayake~\cite{Ramijcar16} the first proof of syntactic cut-elimination for INS calculi was presented by inducing the existing results from another well-known formalism called the \textit{labelled sequent calculus} (introduced in Section~\ref{sec_prelim}). In particular, it was shown that labelled sequent derivations restricted to labelled sequents that are `almost treelike' (labelled tree sequents with equality (LTSE) in the formal terminology of this paper) correspond exactly to indexed nested sequents. This is the technical machinery observed by Fitting~\cite{Fit15} as `a significant different direction' in the study of indexed nested sequents.
The first INS calculi for intermediate logics were introduced using this method, and also INS calculi for the Geach logics omitted in~\cite{Fit15} i.e. Geach axioms with $i+k=0$.
Incidentally, we are not aware of any existing NS calculi for intermediate logics.
Subsequently, Marin and Stra\ss burger~\cite{MarStr17} presented a syntactic proof of cut-elimination for classical and intuitionistic modal logics.

The present work extends~\cite{Ramijcar16} with extensive proof details (omitted from the conference version due to space constraints), and further discussion and examples.

The only other syntactic proof of cut-elimination for indexed nested sequents at present is Marin and Stra\ss burger~\cite{MarStr17}. Let us briefly compare the \textit{procedural} differences in that approach with Ramanayake~\cite{Ramijcar16}. 
Note that this is a comment on the method of argument (and hence a relevant consideration for any further work), rather than on the final cut free derivation which is essentially the same in both cases because both arguments eliminate the cut rule via Gentzen-style reductions.
Since~\cite{Ramijcar16} is a \textit{specialisation} of the existing labelled sequent proof from Negri~\cite{Neg05}, the steps and lemmata that are required are the same. The key thing to check in each lemma is that those labelled sequent transformations preserve LTSE-derivations (i.e. when the input is a LTSE-derivation, the output needs to be a LTSE-derivation too). There is one key point where the proof in~\cite{Neg05} does not preserve LTSE-derivations: that is the substitution lemma. The solution employed is to prove a more nuanced but still adequate substitution lemma (Lem.~\ref{lem-weaksubs}). Note that some care was also required with the weakening lemma (Lem.~\ref{lem_wk}). Notwithstanding these points, we believe that the reader who is familiar with the well-established proof in~\cite{Neg05} would be on familiar ground throughout. In our opinion, this is a major attraction of the specialisation approach. Meanwhile, in~\cite{MarStr17}, the proof of syntactic cut-elimination for nested sequents is \textit{extended} to indexed nested sequents. While the former provides a guide, indexed nested sequents are more complicated so novelty and care are required to handle and check the additional complications. For example, \cite{MarStr17}~proceeds by constructing a variant calculus (replacing the standard initial sequents with a new version where the positive and negative propositional variable may occur at different nestings so long as they have the same index; making the $\diam$-rule non-local) to facilitate their argument. 

\section{Preliminaries}\label{sec_prelim}

The set of natural numbers is denoted by~$\mathbb{N}$.
We assume a set $\{p,q,r,\ldots\}$ of propositional variables. A formula in the language of classical or intermediate logic is either a propositional variable or the logical constants~$\bot, \top$ or has the form $A\star B$ where~$A$ and~$B$ are formulae and $\star\in\{\lor,\land,\imp\}$. Formulae in the language of modal logic can also have the form~$\Box A$ and~$\diam A$ where~$A$ is a formula.
The \textit{size} of a formula is the sum of the number of connectives and propositional variables and logical constants it contains.
\subsection*{Notation}

We employ the following naming convention for variables.
\begin{itemize}
\item $A,B,C$ denote formulae. $p$ denotes a propositional variable.

\item $X,Y,U,V,P,Q,L,M,S,T$ are used to denote (labelled) formula multisets. The exception is with NS and INS:
\begin{itemize}
\item In a NS~$X\DC Y$, $X$~denotes a formula multiset and $Y$~denotes a NS. 
\item In a INS~$X\DC Y$, $X$~denotes a formula multiset and $Y$~denotes a INS. 
\end{itemize}

\item $x,y,z,u,v,e,\alpha,\beta$ are used for state variables in the LS.

\item $a,b,c,d$ are used for indices in the INS.

\item $\Gamma\{ \,\}$ is an INS context. Thus $\Gamma\{ X\DC^{a} Y\}$ is an INS containing the INS~$X\DC^{a} Y$.

\item $\mathcal{R}$ and~$\mathcal{E}$ are used to denote multisets of relation terms and equality terms respectively.
\end{itemize}
Now we will present the formal definitions and relevant results for the formalisms introduced in the previous section.

\subsection{Labelled sequents}

The labelled sequent formalism~\cite{Fit83,Min97} extends the traditional sequent by the prefixing of \textit{state variables} to formulae (e.g. $x:A\imp B$ and~$y:p$) and the inclusion of terms for the binary predicate~$R$. 
For the formal presentation, assume that we have at our disposal an infinite set $\mathbb{SV}=\{ x_{1},x_{2},\ldots\}$ of
\textit{state variables} disjoint from the set of propositional
variables. A \textit{labelled formula} has the form $x:A$ where~$x\in\mathbb{SV}$ and~$A$ is a formula.  If $X=\{ A_{1},\ldots
A_{n} \}$ is a formula multiset, then $x:X$ denotes the multiset $\{
x:A_{1},\ldots,x:A_{n} \}$ of labelled formulae. Notice that if the
formula multiset~$X$ is empty, then the labelled formula multiset $x:X$ is also empty.
A \textit{relation mset}~$\mathcal{R}$ is a multiset of relation terms $Rxy$ ($x,y\in\mathbb{SV}$).
An \textit{equality mset}~$\mathcal{E}$ is a multiset of equality terms $x=y$ ($x,y\in\mathbb{SV}$).
An \textit{atomic term} is either a relation or equality term.
\begin{defi}[labelled sequent]
A \emph{labelled sequent} (LS) has the form $\REL,X\DC Y$ where~$\mathcal{R}$ is a relation mset, $\mathcal{E}$ is an equality mset and~$X$ and~$Y$ are multisets of labelled formulae.
\end{defi}
Intuitively, the labelled sequent can be used to express restricted first-order statements about formulae.
A labelled sequent can also be viewed as a directed graph whose nodes are decorated by traditional sequents~\cite{Res06}.
\begin{defi}[graph defined by relation mset]\label{def-R-graph}
The graph~$\g(\mathcal{R})$ defined by relation mset~$\mathcal{R}$ is the directed graph whose nodes are the state variables in~$\mathcal{R}$ and $x\rightarrow y$ is a directed edge in~$\g(\mathcal{R})$ iff $Rxy\in\mathcal{R}$.
\end{defi}

\begin{defi}[rule, calculus derivation, height]
A \emph{rule} is a non-empty sequence of LSs, typically written as $(s_{1},\ldots,s_{N}/s_{N+1})$. The sequent~$s_{N+1}$ is called the \emph{conclusion} of the rule, the remaining LSs are called the \emph{premises} of the rule. 
If~$N=0$ then the rule is called an \emph{initial LS}. A calculus consists of a finite set of rules. A \emph{LS calculus} is a finite set of LS rules.
A \emph{derivation} in the calculus is defined recursively as either an initial LS or the object obtained by applying a rule~$\rho$ in the calculus to smaller derivations whose bottommost LSs (`endsequents') are legal premises of~$\rho$. 

The \emph{height} of a derivation is the number of rules on its longest branch (viewing the derivation as a tree whose nodes are LSs and the root is the endsequent).
\end{defi}
These definitions are standard and apply to the other types of sequents that we introduce in this paper by replacing ``LS" with the desired sequent type.

Negri~\cite{Neg05} has presented a method for generating cutfree and contraction-free LS calculi for the large family of modal logics whose Kripke semantics are defined by geometric (first-order) formulae.
The proof of cut-elimination is general in the sense that it applies uniformly to every modal logic defined by geometric formulae.
This result has been extended to labelled sequent calculi for intermediate and other non-classical logics~\cite{DycNeg12} and indeed to arbitrary first-order formulae~\cite{DycNeg15} via the addition of constants.
See also Vigan\`{o}~\cite{Vig00} where non-classical logics with semantics defined by Horn formulae are investigated using cutfree labelled calculi introduced therein.

While the primary LS calculus~$\text{\textbf{G3K}}$ in~\cite{Neg05} does not contain equality terms, the extension of that calculus to a theory of equality is presented there (under the name~$\text{\textbf{G3K-Eq}}$). The LS calculus~$\labK$ Fig.~\ref{fig_labK} is obtained from the latter. We have added a new rule~(ls-sc) in anticipation of the INS rule~(sc) in~$\INSK$ that we wish to simulate. The rules~$\text{Repl}_{R1}$ and~$\text{Repl}_{R2}$ from~$\text{\textbf{G3K-Eq}}$ have been omitted. In this paper with LS sequents whose relation mset is a tree: then the former rule can never be applied and the the latter rule is subsumed by~(ls-sc).

\begin{figure}

\begin{small}
\begin{tabular}{cc}
\AxiomC{}
\RightLabel{(init-$\bot$)}
\UnaryInfC{$\REL,x:\bot,X\DC  Y $}
\DisplayProof
&
\AxiomC{}
\RightLabel{(init)}
\UnaryInfC{$\REL,x:p,X\DC  Y , x:p$}
\DisplayProof
\\[1em]
\AxiomC{$\REL,x:A,X\DC  Y $}
\AxiomC{$\REL,x:B,X\DC  Y $}
\RightLabel{(${\lor}$l)}
\BinaryInfC{$\REL,x:A\lor B, X\DC  Y $}
\DisplayProof
&
\AxiomC{$\REL,X\DC  Y ,x:A,x:B$}
\RightLabel{(${\lor}$r)}
\UnaryInfC{$\REL,X \DC  Y ,x:A\lor B$}
\DisplayProof
\\[1.5em]
\AxiomC{$\REL,x:A,x:B,X\DC  Y $}
\RightLabel{(${\land}$l)}
\UnaryInfC{$\REL,x:A\land B,X\DC  Y $}
\DisplayProof
&
\AxiomC{$\REL,X\DC  Y ,x:A$}
\AxiomC{$\REL,X\DC  Y ,x:B$}
\RightLabel{(${\land}$r)}
\BinaryInfC{$\REL,X\DC  Y ,x:A\land B$}
\DisplayProof
\\[1.5em]
\AxiomC{$\REL,X\DC  Y ,x:A$}
\AxiomC{$\REL,x:B,X\DC  Y $}
\RightLabel{(${\imp}$l)}
\BinaryInfC{$\REL,x:A\imp B,X\DC  Y $}
\DisplayProof
&
\AxiomC{$\REL,x:A,X\DC  Y ,x:B$}
\RightLabel{(${\imp}$r)}
\UnaryInfC{$\REL,X\DC  Y ,x:A\imp B$}
\DisplayProof
\\[1.5em]
\AxiomC{$\REL[,Rxy],x:\Box A,y:A,X\DC  Y $}
\RightLabel{(${\Box}$l)}
\UnaryInfC{$\REL[,Rxy],x:\Box A,X\DC  Y $}
\DisplayProof
&
\AxiomC{$\REL[,Rxv],X\DC  Y ,v:A$}
\RightLabel{(${\Box}$r)}
\UnaryInfC{$\REL,X\DC Y ,x:\Box A$}
\DisplayProof
\\[1.5em]
\AxiomC{$\REL[,Rxv],v:A,X\DC  Y $}
\RightLabel{(${\diam}$l)}
\UnaryInfC{$\REL,x:\diam A,X\DC Y $}
\DisplayProof
&
\AxiomC{$\REL[,Rxy],X\DC  Y , x:\diam A, y:A$}
\RightLabel{(${\diam}$r)}
\UnaryInfC{$\REL[,Rxy],X\DC  Y , x:\diam A$}
\DisplayProof
\\[1.5em]
\AxiomC{$\REL,x=y,x:A,y:A,X\DC Y $}
\RightLabel{(rep-l)}
\UnaryInfC{$\REL,x=y,x:A,X\DC Y $}
\DisplayProof
&
\AxiomC{$\REL,x=y,X\DC Y ,x:A,y:A$}
\RightLabel{(rep-r)}
\UnaryInfC{$\REL,x=y,X\DC Y ,x:A$}
\DisplayProof
\\[1.5em]
\AxiomC{$\REL, x=x,X\DC Y $}
\RightLabel{(eq-ref)}
\UnaryInfC{$\REL,X\DC Y $}
\DisplayProof
&
\AxiomC{$\REL, x=z, x=y, y=z, X\DC Y $}
\RightLabel{(eq-trans)}
\UnaryInfC{$\REL, x=y, y=z,X\DC Y $}
\DisplayProof
\\[1.5em]
\multicolumn{2}{c}{
\AxiomC{$\REL[,Rxy,Ruv],x=u,y=v,X\DC Y $}
\RightLabel{(ls-sc) $v$ not in conclusion}
\UnaryInfC{$\REL[,Rxy],x=u,X\DC Y $}
\DisplayProof
}
\end{tabular}
\end{small}

\caption{The labelled sequent calculus~$\labK$. (${\Box}$r), (${\diam}$l) and (ls-sc) have the variable restriction: $v$~does not appear in the conclusion.}
\label{fig_labK}
\end{figure}

Here is a derivation of the normality axiom in~$\labK$.
\begin{equation}\label{lab-der-normal-axiom}
\text{
\AxiomC{$Rxy,x:\Box A,y:A\DC y:B, y:A$}
\AxiomC{$Rxy,y:B,x:\Box A,y:A\DC y:B$}
\RightLabel{(${\imp}$l)}
\BinaryInfC{$Rxy,x:\Box A,y:A,x:\Box(A\imp B),y:A\imp B\DC y:B$}
\RightLabel{($\Box$l), ($\Box$l)}
\UnaryInfC{$Rxy,x:\Box A,x:\Box(A\imp B)\DC y:B$}
\RightLabel{($\Box$r)}
\UnaryInfC{$x:\Box A,x:\Box(A\imp B)\DC x:\Box B$}
\RightLabel{(${\imp}$r)}
\UnaryInfC{$x:\Box(A\imp B)\DC x:(\Box A\imp \Box B)$}
\RightLabel{(${\imp}$r)}
\UnaryInfC{$\DC x:\Box(A\imp B)\imp (\Box A\imp \Box B)$}
\DisplayProof
}
\end{equation}

\begin{defi}[geometric formula]
A \emph{geometric formula} is a formula in the first-order language (binary relations~$R$, ${=}$) of the following form where the~$P_{i}$ are atomic formulae and~$\hat{Q}_{j}$ is a conjunction $Q_{j1}\land \ldots\land Q_{jk_{j}}$ of atomic formulae:
\begin{equation}\label{eq_geo_axiom_eq}
\forall\bar{z}(P_{1}\land\ldots\land P_{m}\imp \exists \bar{x}(\hat{Q}_{1}\lor\ldots\lor \hat{Q}_{n}))
\end{equation}
\end{defi}

\begin{thm}[Negri]\label{thm_Negri}
Let~$L$ be a modal logic of Kripke frames satisfying the geometric formulae~$\{\alpha_{i}\}_{i\in I}$. Then $\labK+\{\rho_{i}\}_{i\in I}$ is a LS
calculus for~$L$ where~$\rho_{i}$ is a structural rule of the form GRS below corresponding to the geometric formulae~(\ref{eq_geo_axiom_eq}).
\[
  \AxiomC{$Q_{1}\sub{\bar{y}_{1}}{\bar{x}_{1}},P,\REL,X\DC Y$}
  \AxiomC{$\ldots$}
  \AxiomC{$Q_{N}\sub{\bar{y}_{N}}{\bar{x}_{N}},P,\REL,X\DC Y$}
  \RightLabel{GRS}
  \TrinaryInfC{$P,X\DC Y$}
  \DisplayProof
\]
Here $Q_{j}=Q_{j1},\ldots,Q_{jk_{j}}$ and $P=P_{1},\ldots,P_{m}$. Moreover, the variables $\bar{y}_{1}\ldots,\bar{y}_{n}$ do not appear in the conclusion.
\end{thm}
The GRS has a side condition that certain state variables do not appear appear in the conclusion. Such a condition is called a \textit{variable restriction} and the variable(s) are called~\textit{eigenvariable(s)}. Indeed, the rules (${\Box}$r) and (${\diam}$l) in~$\labK$ have each a single eigenvariable.

It is well-known that the basic normal modal logic~$K$ consists precisely of those formulae that are valid on all Kripke frames. By the Sahlqvist correspondence and completeness theorems, every Sahlqvist formula~$A_{i}$---i.e. a formula of a certain syntactic form---has a corresponding first-order formula~$\alpha_{i}$ such that the finite axiomatic extension~$K+\{A_{i}\}_{i\in I}$ consists precisely of those formulae that are valid on Krikpe frames satisfying the \textit{frame property}~$\land_{i\in I}\alpha_{i}$. For example, the first-order correspondents of the Sahlqvist formulae~$\Box p\imp p$ and~$\Box p\imp \Box\Box p$ are respectively $\forall x(Rxx)$ (reflexivity) and $\forall x\forall u\forall v(Rxu\land Ruv\imp Rxv)$ (transitivity). Then~$K4T$ i.e. $K+\{\Box p\imp p, \Box p\imp \Box\Box p\}$ is the logic of reflexive and transitive Kripke frames in the sense that $A\in K4T$ iff $A$ is valid on every reflexive and transitive Kripke frame. These results are well-known and we refer the reader to a standard reference on modal logic for the details (see e.g.~\cite{BlaRijVen01}).
Define:
\begin{align*}
&\RR^{0}xy:=\emptyset	&& \RR^{n+1}xy :=\{Rxy_{1},Ry_{1}y_{2},\ldots,Ry_{n-1}y_{n},Ry_{n}y\}	\\
&\RRc^{0}xy:=\top 		&& \RRc^{n+1}xy:=Rxy_{1}\land Ry_{1}y_{2}\land \ldots\land Ry_{n-1}y_{n}\land Ry_{n}y
\end{align*}
E.g. $\RR^{2}xy=\{Rxy_{1},Ry_{1}y\}$ and $\RRc^{2}=Rxy_{1}\land Ry_{1}y$.

Let $\bar{y}=y_{1},\ldots,y_{h-1},y$; $\bar{z}=z_{1},\ldots,z_{j-1},z$; $\bar{u}=u_{1},\ldots,u_{i-1},u$ and $\bar{v}=v_{1},\ldots,v_{k-1},v$
and let~$\lambda$ be the function that returns the last element of a non-empty sequence.
We write~$\forall\bar{y}$ as shorthand for $\forall y_{1}\ldots\forall y_{h-1}\forall y$;

The first-order correspondent~$f(h,i,j,k)$ of the Geach axiom~$G(h,i,j,k)$ is well-known:
\begin{equation}\label{Geach-fo}
\forall x\forall \bar{y}\forall \bar{z}\left(\RRc^{h}xy\land \RRc^{j}xz\longrightarrow \exists \bar{u}\bar{v}\left(\RRc^{i}\lambda(x\bar{y})u\land \RRc^{k}\lambda(x\bar{z})v\land \lambda(x\bar{y}\bar{u})=\lambda(x\bar{z}\bar{v})\right)\right)
\end{equation}
By inspection, every~$f(h,i,j,k)$ is a geometric formula.
The corresponding \textit{Geach (labelled sequent) structural rule} is given below. Note that the variables~$u_{1},\ldots,u_{i},u,v_{1},\ldots,v_{k},v$ do not appear in the conclusion.
\begin{equation}\label{GRS_Geach}
\text{
\AxiomC{$\mathcal{R},\RR^{h}xy,\RR^{j}xz,\RR^{i}\lambda(x\bar{y})u,\RR^{k}\lambda(x\bar{z})v,\lambda(x\bar{y}\bar{u})=\lambda(x\bar{z}\bar{v}),\mathcal{E}, X \DC Y$}
\RightLabel{$\L(h,i,j,k)$}
\UnaryInfC{$\mathcal{R},\RR^{h}xy,\RR^{j}xz,\mathcal{E}, X\DC Y$}
\DisplayProof
}
\end{equation}
Here are some examples of Geach formulae and their corresponding frame conditions:
\begin{align*}
&\text{(ref)}	&& \Box p\imp p 			&& \forall x(\top\rightarrow \exists u(Rxu\land u=x))\text{ so $i=1$ and $h,j,k=0$} \\
&\text{(trans)}	&& \Box p\imp\Box\Box p 	&& \forall x \forall z_{1}\forall z(Rxz_{1}\land Rz_{1}z\rightarrow \exists u(Rxu\land u=z))\text{ so $i=1$, $j=2$; $h,k=0$}\\
&\text{(confl)}	&& \diam\Box p\imp\Box\diam p && \forall x \forall y \forall z(Rxy\land Rxz\rightarrow \exists u \exists v(Ryu\land Rzv\land u=v)\text{ so $h,i,j,k=1$}
\end{align*}
From the Sahlqvist correspondence and completeness theorems (see~\cite{BlaRijVen01}) we have that the modal logic defined by the set~$\{f(h_{s},i_{s},j_{s},k_{s})\}_{s\in S}$
is precisely the modal logic $K+\{G(h_{s},i_{s},j_{s},k_{s})\}_{s\in S}$. 
\begin{exa}\label{eg-conf}
The first-order correspondent of~$G(1,1,1,1)=\diam\Box p\imp\Box\diam p$ is
\[
\forall xyz(Rxy\land Rxz\imp \exists uv(Ryu\land Rzv\land u=v))\quad\quad\quad\quad\quad
\]
Here is the corresponding Geach structural rule. The variables~$u$ and~$v$ do not appear in the conclusion:
\begin{align}\label{LTSE_Geach_rule}
\text{
\AxiomC{$\mathcal{R},Rxy,Rxz,Ryu,Rzv,\mathcal{E},u=v,X\DC Y$}
\RightLabel{$\L(1,1,1,1)$}
\UnaryInfC{$\mathcal{R},Rxy,Rxz,\mathcal{E},X\DC Y$}
\DisplayProof
}
\end{align}
Then~$\labK+\L(1,1,1,1)$ is a calculus for $K+\diam\Box p\imp\Box\diam p$.
\end{exa}

\noindent\textbf{Notation.} Let $\labKst$ denote some extension of~$\labK$ by Geach structural rules.

\subsection{Labelled tree sequents}

Labelled tree sequents~\cite{GorRam12AIML} are essentially LS whose relation mset defines a directed tree.
Labelled tree sequent calculi are isomorphic~\cite{GorRam12AIML} up to state variable names to a nested sequent (an isomorphism with prefixed tableaux has also been shown~\cite{Fit12}).
For example, the NS~(\ref{NSone}) and~(\ref{NStwo}) correspond to the labelled tree sequents\par
\begin{scriptsize}
\begin{align*}
&\underbrace{Rxy,Rxz,Ryu,Ryv}_{\text{relation mset}},x:X, y:P, z:S, u:U, v:L\DC x:Y, y:Q, z:T, u:V, v:M \\
&\overbrace{Rxy,Ryu,Ryv},x:X,y:U, u:L, v:S\DC x:Y,y:V, u:M, v:T
\end{align*}\par
\end{scriptsize}
The isomorphism will be more transparent to the reader if he/she consults the trees that we presented following (\ref{NSone}) and~(\ref{NStwo}). In particular, the relation msets above define those trees and the labelled formulae multisets specify how to decorate the tree.
This was used in~\cite{GorRam12AIML} to answer a question in~\cite{Pog09RSL} concerning the relationship between two distinct proof calculi.
We give a formal exposition below.
\begin{defi}[treelike]
A non-empty relation mset~$\mathcal{R}$ is \emph{treelike} if the directed graph defined by~$\mathcal{R}$ is a tree (i.e. it is rooted, irreflexive and its underlying undirected graph has no cycle).
\end{defi}
\begin{exa}
For each of the following relation msets, the graphs defined by these sets appear above it.
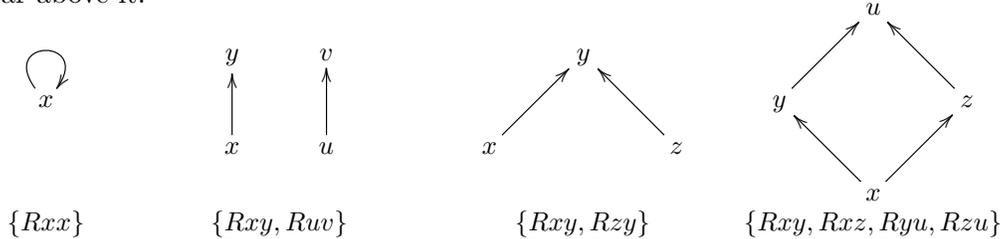
\begin{center}
\vspace{-0.75em}
\begin{small}
\begin{tabular}{c@{\hspace{2em}}c@{\hspace{2em}}c@{\hspace{2em}}c}
\begin{minipage}{3cm}
\xymatrix{
 x \ar@(ul,ur)
}
\end{minipage}
&
\begin{minipage}{3.5cm}
\xymatrix{
& y 	&	 v &\\
& x \ar[u]	&	 u \ar[u]&
}
\end{minipage}
&
\begin{minipage}{3.5cm}
\xymatrix{
		&  y 		&\\
 x \ar[ur] 	& 		&  z \ar[ul]
}
\end{minipage}
&
\begin{minipage}{3.5cm}
\xymatrix{
		&  u  		&		\\
 y \ar[ur] 	&		&  z \ar[ul]\\
		&  x \ar[ul]\ar[ur] &
}
\end{minipage}
\\
$\{ Rxx\}$ & $\{ Rxy,Ruv\}$ & $\{Rxy,Rzy\}$ & $\{Rxy,Rxz,Ryu,Rzu\}$
\end{tabular}
\end{small}
\medskip
\end{center}
None of the above relation msets are treelike because the graphs defined by their relation msets are not trees.
From left-to-right, graph~1 contains a reflexive state; graph~2 and graph~3 are not rooted. 
Frame~4 is not a tree because the underlying undirected graph contains a cycle.
\end{exa}

\begin{defi}[labelled tree sequent]\label{def_LTS}
A \emph{labelled tree-sequent} (LTS) is a labelled sequent of the form $\REL,X\DC Y$ where $\mathcal{E}=\emptyset$ and:
\begin{enumerate}
\item\label{LTS1} if $\mathcal{R}\neq\emptyset$ then $\mathcal{R}$~is treelike and every state variable~$x$
  that occurs in $X\cup Y$ occurs in~$\mathcal{R}$.

\item\label{LTS2} if $\mathcal{R}=\emptyset$ then every label in~$X$ and~$Y$ is the same.
\end{enumerate}
\end{defi}

\noindent Some examples of LTS: 
\begin{align*}
\quad x:A\DC x:B && \phantom{A}\DC y:A && Rxy,Rxz,x:A\DC y:B
\end{align*}
A state variable may occur in the relation mset and not in the $X,Y$ multisets
(e.g.~$z$ above far right). 
Below are \textit{not} $\LTS$ (assume no two in $x,y,z$ identical).\par
\begin{align*}
x:A\DC x:B,z:C && Rxy,x:A\DC z:B && Rxy,Ryz,Rxz\DC\phantom{A}
\end{align*}
From left-to-right above, the first labelled sequent is not an $\LTS$ because the relation mset is 
empty and yet two distinct state variables~$x$ and~$z$ occur in the sequent, violating condition Def.~\ref{def_LTS}(\ref{LTS2}). The next sequent violates Def.~\ref{def_LTS}(\ref{LTS1}) because~$z$
does not appear in the relation mset. The final sequent is not an~$\LTS$ because the relation mset is not treelike.

\subsection{Labelled tree sequents with equality}

\begin{defi}[labelled tree sequent with equality]\label{def_LTSE}
A \emph{labelled tree-sequent with equality} ($\LTSE$) is a labelled sequent of the form $\REL,X\DC Y$ where: 
\begin{enumerate}
\item\label{enum_def_LTSEi} if $\mathcal{R}\neq\emptyset$ then $\mathcal{R}$~is treelike and every state variable in $X,Y$ and~$\mathcal{E}$ occurs in~$\mathcal{R}$.

\item\label{enum_def_LTSEii} if $\mathcal{R}=\emptyset$ then every label in $X,Y$ and~$\mathcal{E}$ is the same.
\end{enumerate}
\end{defi}
\noindent Clearly every~$\LTS$ is an $\LTSE$. Each of the following is an~$\LTSE$:
\begin{align*}
Rxy,x=y,x:A\DC y:B && y=y\DC y:A && Rxy,Rxz,y=z,x:A\DC z:B
\end{align*}
The following are \textit{not} $\LTSE$ (assume that no two in $x,y$ and~$z$ are identical).
\begin{align*}
x=z,x:A\DC x:B && Rxy,y=z,x:A\DC y:B && Rxy,Ryz,Rxz,x=y,x=z\DC\phantom{A}
\end{align*}
From left-to-right above, the first labelled sequent is not an $\LTSE$ because the relation mset is 
empty and yet the sequent contains more than one label. The next sequent violates Def.~\ref{def_LTSE}(\ref{enum_def_LTSEii}) because~$z$ does not appear in the relation mset.
The final sequent is not an~$\LTSE$ because the relation mset is not treelike.

\subsection{Indexed nested sequents}


The indexed nested sequent calculus~$\INSK$~\cite{Fit15} for~$K$ is given in Fig.~\ref{fig_INSK}.
We write $\Gamma\ob X\DC^{a} Y\bl\context\cb$ to denote that the INS~$\Gamma$ contains the INS $X\DC^{a} Y\bl\context$. In the latter, $X$~is a multiset of formulae and~$Y$ is an INS.
Also $\Gamma\ob X\DC^{a} Y\bl\context\cb\ob U\DC^{b} V\bl\contexta\cb$ denotes an INS~$\Gamma$ containing two distinct INS: $X\DC^{a} Y\bl\context$ and $U\DC^{b} V\bl\contexta$.


Fitting's Geach scheme~\cite[Sec.~8]{Fit15} yields an INS rule~$\I(h,i,j,k)$ corresponding to $G(h,i,j,k)$ when $i+k>0$. 
Here is the INS rule~$\I(1,1,1,1)$, where index~$c$ does not appear in the conclusion.
\begin{equation}\label{INS_Geach_rule}
\text{
\AxiomC{$\Gamma\ob 	[X\DC^{a} Y\bl\context, [\DC^{c}]], [U\DC^{b} V\bl\contexta, [\DC^{c}]]	\cb$}
\RightLabel{$\I(1,1,1,1)$}
\UnaryInfC{$\Gamma\ob 	[X\DC^{a} Y\bl\context], [U\DC^{b} V\bl\contexta]	\cb$}
\DisplayProof
}
\end{equation}
Here is a derivation of $\DC G(1,1,1,1)$ i.e. $\DC^{0}\diam\Box p\imp\Box\diam p$ in $\INSK+\I(1,1,1,1)$.
\[
  \AxiomC{$\DC^{0}, [\DC^{1} \diam p, [\DC^{3} p]], [\Box p\DC^{2}, [p\DC^{3} p]]$}
  \RightLabel{(fc-r)}
  \UnaryInfC{$\DC^{0}, [\DC^{1} \diam p, [\DC^{3} p]], [\Box p\DC^{2}, [p\DC^{3}]]$}
  \RightLabel{($\Box$l)}
  \UnaryInfC{$\DC^{0}, [\DC^{1} \diam p, [\DC^{3} p]], [\Box p\DC^{2}, [\DC^{3}]]$}
  \RightLabel{($\diam$r)}
  \UnaryInfC{$\DC^{0}, [\DC^{1} \diam p, [\DC^{3}]], [\Box p\DC^{2}, [\DC^{3}]]$}
  \RightLabel{\I(1,1,1,1)}
  \UnaryInfC{$\DC^{0}, [\DC^{1} \diam p], [\Box p\DC^{2}]$}
  \RightLabel{($\diam$l)}
  \UnaryInfC{$\diam\Box p\DC^{0}, [\DC^{1} \diam p]$}
  \RightLabel{($\Box$r)}
  \UnaryInfC{$\diam\Box p\DC^{0} \Box\diam p$}
  \RightLabel{(${\imp}$r)}
  \UnaryInfC{$\DC^{0}\diam\Box p\imp\Box\diam p$}
  \DisplayProof
\]
In contrast to~\cite{Fit15} where INS rules are not given for Geach axioms of the form~$G(h,0,j,0)$, in this work we give INS calculi for \textit{all} Geach axioms using a rule scheme (see Sec.~\ref{translation}) that coincides with Fitting's scheme for $i+k>0$.


\begin{figure}
\begin{center}
\begin{small}
\begin{flalign*}
  &&
\AxiomC{}
\RightLabel{(init-$\bot$)}
\UnaryInfC{$\Gamma\ob \bot,X\DC^{a} Y\bl\context\cb$}
\DisplayProof
&&
\AxiomC{}
\RightLabel{(init)}
\UnaryInfC{$\Gamma\ob p,X\DC^{a} Y,p\bl\context\cb$}
\DisplayProof
&&
\AxiomC{$\Gamma\ob \Box A,X\DC^{a} Y\bl\context, [A, U\DC^{b} V\bl\contexta]\cb$}
\RightLabel{(${\Box}$l)}
\UnaryInfC{$\Gamma\ob \Box A,X\DC^{a} Y\bl\context, [U\DC^{b} V\bl\contexta]\cb$}
\DisplayProof
&&
\end{flalign*}
\smallskip
\begin{flalign*}
  &&
\AxiomC{$\Gamma\ob X\DC^{a} Y,A,B\bl\context\cb$}
\RightLabel{(${\lor}$r)}
\UnaryInfC{$\Gamma\ob X\DC^{a} Y,A\lor B\bl\context\cb$}
\DisplayProof
&&
\AxiomC{$\Gamma\ob X,A,B\DC^{a} Y\bl\context\cb$}
\RightLabel{(${\land}$l)}
\UnaryInfC{$\Gamma\ob X,A\land B\DC^{a} Y\bl\context\cb$}
\DisplayProof
&&
\AxiomC{$\Gamma\ob X\DC^{a} Y,A\bl\context\cb$}
\AxiomC{$\Gamma\ob X\DC^{a} Y,B\bl\context\cb$}
\RightLabel{(${\land}$r)}
\BinaryInfC{$\Gamma\ob X\DC^{a} Y,A\land B\bl\context\cb$}
\DisplayProof
&&
\end{flalign*}
\smallskip
\begin{flalign*}
  &&
\AxiomC{$\Gamma\ob A,X\DC^{a} Y\bl\context \cb \ob A,U\DC^{a} V\bl\contexta\cb$}
\RightLabel{(fc-l)}
\UnaryInfC{$\Gamma\ob A,X\DC^{a} Y\bl\context \cb \ob U\DC^{a} V\bl\contexta\cb$}
\DisplayProof
&&
\AxiomC{$\Gamma\ob A,X\DC^{a} Y,B\bl\context\cb$}
\RightLabel{(${\imp}$r)}
\UnaryInfC{$\Gamma\ob X\DC^{a} Y,A\imp B\bl\context\cb$}
\DisplayProof
&&
\AxiomC{$\Gamma\ob X\DC^{a} Y\bl\context, [\quad\DC^{b} A]\cb$}
\RightLabel{(${\Box}$r)}
\UnaryInfC{$\Gamma\ob X\DC^{a} Y\bl\context, \Box A\cb$}
\DisplayProof
&&
\end{flalign*}
\smallskip
\begin{flalign*}
  &&
\AxiomC{$\Gamma\ob A,X\DC^{a} Y\bl\context\cb$}
\AxiomC{$\Gamma\ob B,X\DC^{a} Y\bl\context\cb$}
\RightLabel{(${\lor}$l)}
\BinaryInfC{$\Gamma\ob A\lor B,X\DC^{a} Y\bl\context\cb$}
\DisplayProof
&&
\AxiomC{$\Gamma\ob X\DC^{a} Y,A\bl\context\cb$}
\AxiomC{$\Gamma\ob B,X\DC^{a} Y\bl\context\cb$}
\RightLabel{(${\imp}$l)}
\BinaryInfC{$\Gamma\ob A\imp B,X\DC^{a} Y\bl\context\cb$}
\DisplayProof
&&
\end{flalign*}
\smallskip
\begin{flalign*}
  &&
\AxiomC{$\Gamma\ob X\DC^{a} Y\bl\context, [A\DC^{b} \quad]\cb$}
\RightLabel{(${\diam}$l)}
\UnaryInfC{$\Gamma\ob \diam A,X\DC^{a} Y\bl\context\cb$}
\DisplayProof
&&
\AxiomC{$\Gamma\ob X\DC^{a} Y, \diam A\bl\context, [U\DC^{b} V, A\bl\contexta]\cb$}
\RightLabel{(${\diam}$r)}
\UnaryInfC{$\Gamma\ob X\DC^{a} Y, \diam A\bl\context, [U\DC^{b} V\bl\contexta]\cb$}
\DisplayProof
&&
\end{flalign*}
\smallskip
\begin{flalign*}
  &&
\AxiomC{$\Gamma\ob X\DC^{a} Y,A\bl\context\cb\ob U\DC^{a} V,A\bl\contexta\cb$}
\RightLabel{(fc-r)}
\UnaryInfC{$\Gamma\ob X\DC^{a} Y,A\bl\context\cb\ob U\DC^{a} V\bl\contexta\cb$}
\DisplayProof
&&
\AxiomC{$\Gamma\left\ob X\DC^{a} Y\bl\context, [P\DC^{b} Q\bl\contextb]\right\cb\left\ob U\DC^{a} V\bl\contexta, [\phantom{a}\DC^{b}\phantom{a}] \right\cb$}
\RightLabel{(sc)}
\UnaryInfC{$\Gamma\left\ob X\DC^{a} Y\bl\context, [P\DC^{b} Q\bl\contextb]\right\cb\left\ob U\DC^{a} V\bl\contexta\right\cb$}
\DisplayProof
&&
\end{flalign*}
\end{small}
\end{center}

\caption{The indexed nested sequent calculus~$\INSK$ for modal logic~$K$. The rules (${\Box}$r) and~(${\diam}$l) have the side condition that index~$b$ does not appear in the conclusion.}
\label{fig_INSK}
\end{figure}


\section{Syntactic cut-elimination for $\LTSE$-derivations}\label{sec_results}

An \textit{LTSE-derivation} is a restricted form of a labelled sequent derivation.
\begin{defi}[$\LTSE$-derivation in~$\labKst$]
An $\LTSE$-derivation in~$\labKst$ is a derivation in that calculus where every sequent is an~$\LTSE$.
\end{defi}
Let $\bar{x}=(x_{1},\ldots,x_{N+1})$ and $\bar{y}=(y_{1},\ldots,y_{N+1})$. Then $\mathcal{R}\sub{\bar{y}}{\bar{x}}$ is obtained by replacing every occurrence of~$x_{i}$ in relation mset~$\mathcal{R}$ with~$y_{i}$ ($1\leq i\leq N+1$). For an equality mset~$\mathcal{E}$ and a labelled formula multiset~$X$ define $\mathcal{E}\sub{\bar{y}}{\bar{x}}$ and $X\sub{\bar{y}}{\bar{x}}$ analogously.
When~$\bar{x}=(x)$ and~$\bar{y}=(y)$ we simply write~$\sub{y}{x}$. In words, $\sub{y}{x}$ is the substitution of every occurrence of~$x$ with~$y$.

We will now demonstrate how Negri's~\cite{Neg05} proof of cut-eliminability (also called cut-admissibility) in LS calculi can be reused to give a proof for the restricted class of LTSE-derivations. The proof in~\cite{Neg05} is not immediately a proof of cut-eliminability for LTSE-derivations it needs to be verified that every step of the cut-eliminability transformation stays within the class of LTSE derivations.

A rule is \emph{invertible} if the premises are derivable whenever the conclusion is derivable. A rule (not necessarily belonging to the calculus) is \emph{admissible} if the conclusion is derivable whenever the premises are derivable. The transformation is \textit{height-preserving} if the height of the final derivation is no greater than the height of the original derivation. Let us recall the key lemmata in~\cite{Neg05}. 
\begin{description}
\item[identity derivation] The sequent $\REL,x:A,X \DC x:A, Y$ for arbitrary formula~$A$ is derivable in~$\labKst$~\cite[Lem. 4.1]{Neg05}.

\item[hp substitution] If $\REL, X \DC Y $  is derivable in~$\labKst$, then there is a derivation of $\REL[\sub{y}{x}]\sub{y}{x}, X \sub{y}{x}\DC Y \sub{y}{x}$ of no greater height~\cite[Lem. 4.3]{Neg05}. 

\item[hp weakening] If $\REL, X \DC Y $ is derivable in~$\labKst$, then there is a derivation of $\mathcal{R},\mathcal{R}',\mathcal{E},\mathcal{E}', X ,X'\DC Y,Y'$ of no greater height~\cite[Lem. 4.4]{Neg05}.

\item[hp invertibility] All the rules of~$\labKst$ are height-preserving invertible~\cite[Lem. 4.11]{Neg05}.

\item[hp contraction] The rules of contraction are height-preserving admissible in~$\labKst$~\cite[Lem. 4.12]{Neg05}.
\end{description}
The substitution property above can transform an LTSE into a non-LTSE (since it operates directly on the relation mset) and hence we require a more nuanced substitution lemma (Lem.~\ref{lem-weaksubs}).
\begin{quote}
Suppose that sequent~$s$ has a LTSE-derivation. Suppose that~$\alpha$ is in the relation mset of~$s$, and that $s\sub{\beta}{\alpha}$ is a LTSE. Then $s\sub{\beta}{\alpha}$ has a LTSE-derivation of the same height.
\end{quote}
The main addition to the proofs of the lemmata above is ensuring that if the given derivation is an LTSE-derivation, then so is the transformed derivation.

\begin{lem}[identity derivation]
The sequent $\REL,x:A, X \DC x:A, Y $ for arbitrary formula~$A$ is LTSE-derivable in~$\labKst$.
\end{lem}
\begin{proof}
Proceed by induction on the size of~$A$.
\end{proof}
This following is called a renaming lemma because the substituting variable is required to be fresh.
\begin{lem}[renaming]\label{lem-renaming}
If $\REL, X \DC Y $ does not contain~$\alpha$ and is LTSE-derivable, then there is a LTSE-derivation of $\REL[\sub{\beta}{\alpha}]\sub{\beta}{\alpha}, X \sub{\beta}{\alpha}\DC Y \sub{\beta}{\alpha}$ of no greater height.
\end{lem}
\begin{proof}
Proceed by induction on the height of the LTSE-derivation~$d$ of $s=\REL, X \DC Y $. Consider the last rule~$r$ in~$d$. If~$\alpha$ is an eigenvariable in the premise of~$r$,  then it cannot occur in~$s$, so~$s\sub{\beta}{\alpha}$ and~$s$ are identical. Otherwise, apply the induction hypothesis to every premise of~$r$ and then apply~$r$ to the desired LTSE-derivation.
\end{proof}

\begin{lem}[LTSE-substitution]\label{lem-weaksubs}
If sequent~$s$ has an LTSE-derivation, $\alpha$ is in the relation mset of~$s$ and $s\sub{\beta}{\alpha}$ is a LTSE, then $s\sub{\beta}{\alpha}$ has a LTSE-derivation of the same height.
\end{lem}
\begin{proof}
Let~$d$ be a derivation of~$s$. Proceed by induction on the height of~$s$.

In the base case, $d$ has height~$1$ and is hence an initial sequent $\mathcal{R},\mathcal{E},x:\bot, X \DC Y $ or $\mathcal{R},\mathcal{E},x:p, X \DC Y ,x:p$. By hypothesis, $\mathcal{R}\sub{\beta}{\alpha}$ defines a tree. The following are then the respective LTSE-derivations.
\begin{align*}
&\mathcal{R}\sub{\beta}{\alpha},\mathcal{E}\sub{\beta}{\alpha},(x:\bot)\sub{\beta}{\alpha}, X \sub{\beta}{\alpha}\DC Y \sub{\beta}{\alpha} \\
&\mathcal{R}\sub{\beta}{\alpha},\mathcal{E}\sub{\beta}{\alpha},(x:p)\sub{\beta}{\alpha}, X \sub{\beta}{\alpha}\DC Y \sub{\beta}{\alpha},(x:p)\sub{\beta}{\alpha}
\end{align*}

For the inductive case, suppose that~$d$ has height~$>1$. Consider the last rule~$r$ in~$d$.

Suppose that~$r$ is a logical rule (the rule introduces a new logical---propositional---connective in the conclusion), or~(sym) or~(trans). Then the rule instance can be written in the following form (we illustrate with a unary rule, a binary rule is analogous).
\[
  \AxiomC{$\mathcal{R},\mathcal{E}, X_{1}\DC Y_{1}$}
  \UnaryInfC{$\mathcal{R},\mathcal{E}, X_{0}\DC Y_{0}$}
  \DisplayProof
\]
If~$r$ has an eigenvariable~$v$ in the premise, first apply Lem.~\ref{lem-renaming} to rename the eigenvariable~$v$ to a variable~$v^{*}$ such that $v^{*}$~does not occur in the conclusion \textit{and} $v^{*}\not\in\{\alpha,\beta\}$.
By assumption, the premise is LTSE-derivable, $\alpha$ occurs in its relation mset and~$\mathcal{R}\sub{\beta}{\alpha}$ defines a tree.
We can thus apply the IH to obtain an LTSE-derivation of
\[
(\mathcal{R}\sub{v^{*}}{v})\sub{\beta}{\alpha},(\mathcal{E}\sub{v^{*}}{v})\sub{\beta}{\alpha},( X_{1}\sub{v^{*}}{v})\sub{\beta}{\alpha}\DC( Y_{1}\sub{v^{*}}{v})\sub{\beta}{\alpha}
\]
Here and in the remainder of \textit{this proof}, for brevity we will append a prime~$\gamma'$ to every variable/term~$\gamma$ as shorthand to mean~$\gamma\sub{\beta}{\alpha}$ (rather than writing~$\sub{\beta}{\alpha}$ explicitly). Thus the above sequent will be denoted
\[
(\mathcal{R}\sub{v^{*}}{v})',(\mathcal{E}\sub{v^{*}}{v})',( X_{1}\sub{v^{*}}{v})'\DC( Y_{1}\sub{v^{*}}{v})'
\]
Now apply~$r$ to obtain
\[
\mathcal{R}',\mathcal{E}', X_{1}'\DC  Y_{1}'
\]
This is the desired LTSE-derivation.

Suppose that the last rule is (rep-l) (the case of (rep-r) is similar).
\[
  \AxiomC{$\REL,x=y,x:A,y:A, X \DC Y $}
  \RightLabel{(rep-l)}
  \UnaryInfC{$\REL,x=y,x:A, X \DC Y $}
  \DisplayProof
\]
By assumption we know that the premise is a LTSE-derivation, $\alpha$~occurs in the relation mset of~$\mathcal{R}$ and~$\mathcal{R}\sub{\beta}{\alpha}$ defines a tree. 
Therefore we can apply the IH to the premise to obtain an LTSE-derivation of the following (recall that `priming' the variables is shorthand to indicate that the substitution has been applied)
\[
\mathcal{R}',\mathcal{E}',x'=y',x':A,y':A, X'\DC Y'
\]
Apply (rep-l) to obtain a LTSE-derivation of $\mathcal{R}',\mathcal{E}',x'=y',x':A,y':A, X'\DC Y'$ as desired.

Suppose that the last rule is (ls-sc); $v$ does not appear in the conclusion:
\[
  \AxiomC{$\REL[,Rxy,Ruv],x=u,y=v, X \DC Y $}
  \RightLabel{(ls-sc)}
  \UnaryInfC{$\REL[,Rxy],x=u, X \DC Y $}
  \DisplayProof
\]
First, apply Lem.~\ref{lem-renaming} to the premise to rename the eigenvariable~$v$ to a variable~$v^{*}$ such that $v^{*}$~does not occur in the conclusion \textit{and} $v^{*}\not\in\{\alpha,\beta\}$. We then obtain
\begin{equation}\label{weaksubs-lssc}
\REL[,Rxy,Ruv^{*}],x=u,y=v^{*}, X \DC Y 
\end{equation}
Since $\mathcal{R},Rxy$ contains~$u$ and because~$\mathcal{R}\sub{\beta}{\alpha},(Rxy)\sub{\beta}{\alpha}$ defines a tree by assumption, it follows that~$u\sub{\beta}{\alpha}$ is a node in this tree.
Appending the prime symbol to denote~$\sub{\beta}{\alpha}$, we write the latter relation set as $\mathcal{R}',Rx'y'$. Since~$v^{*}$ does not occur in this relation mset, and since adding an edge from a node in the tree to a new node preserves the tree property, it follows that~$\mathcal{R}',Rx'y',Ru'v^{*}$ defines a tree as well. Thus~(\ref{weaksubs-lssc}) is LTSE-derivable, its relation mset contains~$\alpha$, and it remains a LTSE under~$\sub{\beta}{\alpha}$. By IH we obtain a height-preserving derivation of the following.
\[
\mathcal{R}',Rx'y',Ru'v^{*},\mathcal{E}',x'=u',y'=v^{*}, X'\DC Y'
\]
Now apply (ls-sc) to obtain $\mathcal{R}',Rx'y',\mathcal{E}',x'=u', X'\DC Y'$. This was the desired LTSE-derivation.

Suppose that the last rule is a Geach structural rule. To minimise the notational clutter, let us demonstrate with the concrete rule in~(\ref{LTSE_Geach_rule}) (the general case is straightforward).
\[
  \AxiomC{$\mathcal{R},Rxy,Rxz,Ryu,Rzv,\mathcal{E},u=v, X \DC Y $}
  \RightLabel{$\L(1,1,1,1)$}
  \UnaryInfC{$\mathcal{R},Rxy,Rxz,\mathcal{E}, X \DC Y $}
  \DisplayProof
\]
Let~$u^{*},v^{*}$ be distinct variables that do not occur in the conclusion \textit{and} $u^{*},v^{*}\not\in\{\alpha,\beta\}$. Proceed as follows.
\begin{align*}
  \AxiomC{$\mathcal{R},Rxy,Rxz,Ryu,Rzv,\mathcal{E},u=v, X \DC Y $}
  \RightLabel{hp renaming}
  \dashedLine
  \UnaryInfC{$\mathcal{R},Rxy,Rxz,Ryu^{*},Rzv^{*},\mathcal{E},u^{*}=v^{*}, X \DC Y $}
  \dashedLine
  \RightLabel{IH}
  \UnaryInfC{$\mathcal{R}',Rx'y',Rx'z',Ry'u',Rz'v^{*},\mathcal{E},u^{*}=v^{*}, X'\DC Y'$}
  \RightLabel{$\L(1,1,1,1)$}
  \UnaryInfC{$\mathcal{R}',Rx'y',Rx'z', X'\DC Y'$}
  \DisplayProof\\[-\normalbaselineskip]\tag*{\qedhere}
\end{align*}
\end{proof}

A relation mset~$\mathcal{R}_{1}$ is a \textit{tree-extension} of~$\mathcal{R}_{0}$ if (i)~the graph~$\g(\mathcal{R}_{1})$ is a tree and (ii)~the graph~$\g(\mathcal{R}_{0})$ is a subtree of~$\g(\mathcal{R}_{1})$. Additionally, if the root of~$\g(\mathcal{R}_{0})$ and the root of~$\g(\mathcal{R}_{1})$ are identical, then~$\mathcal{R}_{1}$ is called a \textit{special tree-extension} of~$\mathcal{R}_{0}$. Those vertices and edges not in the subtree are called the \textit{new vertices and edges}.
\begin{lem}\label{lem-te}
Let~$\mathcal{R}_{1}$ be a tree-extension of~$\mathcal{R}_{0}$, and let~$\mathcal{R}^{*}$ be a special tree-extension of~$\mathcal{R}_{0}$ such that no new vertex of~$\g(\mathcal{R}^{*})$ occurs in~$\g(\mathcal{R}_{1})$.
Then $\mathcal{R}_{1},\mathcal{R}^{*}$ is a tree-extension of~$\mathcal{R}_{0}$.
\end{lem}
\begin{proof}
Let~$t_{0}$ denote~$\g(\mathcal{R}_{1})$. There must be a sequence $(u_{1},v_{1}),\ldots,(u_{N},v_{N})$ of new edges---witnessing that~$\g(\mathcal{R}^{*})$ is a special tree-extension of~$\g(\mathcal{R}_{0})$---added successively to~$t_{0}$ (let~$t_{j}$ be the connected directed graph obtained after that~$j^{\text{th}}$ addition) such that for every~$j$ ($1\leq j\leq N$):
\begin{enumerate}
\item\label{te-one} $t_{N}$ is $\g(\mathcal{R}_{1},\mathcal{R}^{*})$

\item\label{te-two} $u_{j}\in t_{j-1}$ and~$v_{j}\not\in t_{j-1}$. By connectedness of~$t_{j}$, one of~$u_{j}$ and~$v_{j}$ must occur in~$t_{j-1}$. Suppose~$v_{j}\in t_{j-1}$: if~$v_{j}\in \g(\mathcal{R}_{0})$ then this would violate the ``special" property, and else if~$v_{j}\in (\g(\mathcal{R}_{1})\cap \g(\mathcal{R}_0)^{c})$ this would violate the ``no new vertex of~$\g(\mathcal{R}^{*})$ (wrt $\g(\mathcal{R}_{0})$) occurs in~$\g(\mathcal{R}_{1})$" assumption.
 
\item\label{te-three} Since~$u_{j}\in t_{j-1}$ and $v_{j}\not\in t_{j-1}$ by~(\ref{te-two}), if~$t_{j-1}$ is a tree containing~$\mathcal{R}_{0}$ as a subgraph, then this holds for~$t_{j}$ too.
\end{enumerate}
Since~$t_{0}$ is a tree containing~$\g(\mathcal{R}_{0})$ as a subtree (because~$\mathcal{R}_{1}$ is a tree extension of~$\mathcal{R}_{0}$) from~(\ref{te-three}) it follows that~$t_{N}$ is a tree containing~$\g(\mathcal{R}_{0})$ as a subtree. Then by~(\ref{te-one}), $\mathcal{R}_{1},\mathcal{R}^{*}$ is a tree-extension of~$\mathcal{R}_{0}$.
\end{proof}

\begin{lem}\label{lem-rules-ste}
In every rule instance of~($\Box$r), (ls-sc) and Geach structural rule such that the conclusion is an LTSE, the premise relation mset(s) are special tree-extensions of the conclusion relation mset.
\end{lem}
\begin{proof}
By assumption, the conclusion relation mset defines a tree. The result then follows from the observation that the relation msets extend the conclusion mset by terms~$Rxy$ such that~$y$ is always an eigenvariable.
\end{proof}

\begin{lem}[height-preserving weakening]\label{lem_wk}
Suppose that~$\mathcal{R}_{0}, X_{0}\DC Y_{0}$ is LTSE-derivable. If $\mathcal{R}_{0},\mathcal{R}',\mathcal{E}_{0},\mathcal{E}', X_{0}, X'\DC Y_{0}, Y'$ is an LTSE, then it has a LTSE-derivation of no greater height.
\end{lem}
\begin{proof}
Let~$d$ be a LTSE-derivation of~$\mathcal{R}_{0}, X_{0}\DC Y_{0}$. Proceed by induction on the height of~$d$.

In the base case, $\mathcal{R}_{0}, X_{0}\DC Y_{0}$ is an initial sequent. This means that $\mathcal{R}_{0},\mathcal{R}',\mathcal{E}_{0},\mathcal{E}', X_{0}, X'\DC Y_{0}, Y'$ is an initial sequent as well.

For the inductive case, suppose that~$d$ has height~$>1$. Consider the last rule~$r$ in~$d$. For every rule that does not contain eigenvariables, the relation and equality msets of the conclusion and the premise(s) are identical. We can abstract the form of such a rule as follows (we illustrate with a unary rule, an analogous argument applies to rules with more premises).
\[
  \AxiomC{$\mathcal{R}_{0},\mathcal{E}_{0}, X_{1}\DC Y_{1}$}
  \UnaryInfC{$\mathcal{R}_{0},\mathcal{E}_{0}, X_{0}\DC Y_{0}$}
  \DisplayProof
\]
That $\mathcal{R}_{0},\mathcal{R}',\mathcal{E}_{0},\mathcal{E}', X_{1}, X'\DC Y_{1}, Y'$ is an LTSE follows from the assumption. Hence by IH the sequent has an LTSE-derivation. Now apply~$r$ to obtain an LTSE-derivation of $\mathcal{R}_{0},\mathcal{R}',\mathcal{E}_{0},\mathcal{E}', X_{0}, X'\DC Y_{0}, Y'$.

If~$r$ does contain eigenvariables---the rules are($\Box$r), (ls-sc) and the Geach structural rules, then we can abstract the form of the rule as follows (we illustrate with a unary rule, an analogous argument applies to rules with more premises).
\[
  \AxiomC{$\mathcal{R}_{1},\mathcal{E}_{1}, X_{1}\DC Y_{1}$}
  \UnaryInfC{$\mathcal{R}_{0},\mathcal{E}_{0}, X_{0}\DC Y_{0}$}
  \DisplayProof
\]
Let~$\bar{e}$ be the set of eigenvariables of~$r$. Apply the hp renaming (Lem.~\ref{lem-renaming}) in the premise to rename the variables~$\bar{e}$ to fresh variables $\bar{e}^{*}$ appearing in neither the premise nor conclusion. 
In this way from the premise we obtain $\mathcal{R}_{1}\sub{\bar{e}^{*}}{\bar{e}},\mathcal{E}_{1}\sub{\bar{e}^{*}}{\bar{e}}, X_{1}\sub{\bar{e}^{*}}{\bar{e}}\DC Y_{1}\sub{\bar{e}^{*}}{\bar{e}}$. By Lem.~\ref{lem-rules-ste}, $\mathcal{R}_{1}\sub{\bar{e}^{*}}{\bar{e}}$ is a special tree extension of~$\mathcal{R}_{0}$ whose new vertices are not in the weakening term~$\mathcal{R}'$. Moreover~$\mathcal{R}'$ is a tree-extension of~$\mathcal{R}_{0}$. From Lem.~\ref{lem-te}, $\mathcal{R}',\mathcal{R}_{1}\sub{\bar{e}^{*}}{\bar{e}}$ is a tree-extension of~$\mathcal{R}_{0}$.
Thus $\mathcal{R}_{1}\sub{\bar{e}^{*}}{\bar{e}},\mathcal{R}',\mathcal{E}_{1}\sub{\bar{e}^{*}}{\bar{e}},\mathcal{E}',X_{1}\sub{\bar{e}^{*}}{\bar{e}},X'\DC Y_{1}\sub{\bar{e}^{*}}{\bar{e}},Y'$ is an LTSE. Apply IH to obtain its LTSE-derivation. Now apply~$r$ to obtain an LTSE-derivation of $\mathcal{R}_{0},\mathcal{R}',\mathcal{E}_{0},\mathcal{E}',X_{0},X'\DC Y_{0},Y'$.
\end{proof}

\begin{lem}[height-preserving invertibility]\label{lem_inv}
If the conclusion of a logical rule in~$\labKst$ has an LTSE-derivation, then its premise has an LTSE-derivation of no greater height.
\end{lem}
\begin{proof}
(sketch) Let~$d$ be an LTSE-derivation of the conclusion of a rule instance in~$\labKst$. Proceed by induction on the height of~$d$.
Consider the last rule~$r$ in~$d$. In every structural rule, every term in the conclusion occurs in the premise. For this reason, the result follows from height-preserving weakening. For the logical rules, apply the IH to the premise of~$r$ and then apply~$r$.
\end{proof}

There are four rules of contraction.
\medskip
\begin{center}
\begin{tabular}{c@{\hspace{2cm}}c}
\AxiomC{$\REL[,Rxy,Rxy], X\DC Y $}
\UnaryInfC{$\REL[,Rxy], X\DC Y $}
\DisplayProof
&
\AxiomC{$\REL,x=y,x=y, X\DC Y $}
\UnaryInfC{$\REL,x=y, X\DC Y $}
\DisplayProof
\\[1.5em]
\AxiomC{$\REL,x:A,x:A, X\DC Y $}
\UnaryInfC{$\REL,x:A, X\DC Y $}
\DisplayProof
&
\AxiomC{$\REL, X\DC Y ,x:A,x:A$}
\UnaryInfC{$\REL,x:A, X\DC Y ,x:A$}
\DisplayProof
\end{tabular}
\end{center}
\medskip
We prove their height-preserving admissibility simultaneously.
\begin{lem}[height-preserving contraction admissibility]\label{lem_ctr}
If the premises of any of the above contraction rules is LTSE-derivable, then there is an LTSE-derivation of the conclusion of no greater height.
\end{lem}
\begin{proof}
Let~$d$ be a LTSE-derivation of a premise of a contraction rule. By induction on the height of~$d$. 
Consider the last rule~$r$ in~$d$. If the contraction is on a relation term or an equality term, then we can use the IH to apply the contraction on the premises and then apply~$r$. 
The remaining case is the admissibility of the formula contraction rules. Suppose that the conclusion of~$d$ contains~$x:A,x:A$. If~$A$ is not principal by~$r$, then once again we can proceed by applying the IH to the premise(s) and then apply~$r$. When one occurrence of~$A$ \textit{is} principal by~$r$, then the idea is to first use height-preserving invertibility (Lem.~\ref{lem_inv}) on the other occurrence of~$A$ followed by the IH and then apply~$r$. When~$r$ is~(${\Box}$r)---the rule ($\diam$l) is similar)---care must be taken with the eigenvariables as shown below.

Suppose that~$d$ concludes as below left. Then the required LTSE-derivation is below right ($y$ and~$z$ are eigenvariables).
\begin{align*}
\begin{tabular}{cc}
\AxiomC{$\REL[,Rxy], X\DC Y ,y:A, x:\Box A$}
\RightLabel{(${\Box}$r)}
\UnaryInfC{$\REL, X\DC Y ,x:\Box A,x:\Box A$}
\DisplayProof
&
\AxiomC{$\REL[,Rxy], X\DC Y ,y:A,x:\Box A$}
\dashedLine
\RightLabel{hp inv}
\UnaryInfC{$\REL[,Rxy,Rxz], X\DC Y ,y:A,z: A$}
\dashedLine
\RightLabel{hp LTSE-subs.}
\UnaryInfC{$\REL[,Rxz], X\DC Y ,z:A,z: A$}
\dashedLine
\RightLabel{IH}
\UnaryInfC{$\REL[,Rxz], X\DC Y ,z:A$}
\RightLabel{(${\Box}$r)}
\UnaryInfC{$\REL, X\DC Y ,x:\Box A$}
\DisplayProof
\end{tabular}\\[-\normalbaselineskip]\tag*{\qedhere}
\end{align*}
\end{proof}
The usual cut rule is given below. 
\[
  \AxiomC{$\mathcal{R}_{1},\mathcal{E}_{1}, X\DC  Y ,x:A$}
  \AxiomC{$\mathcal{R}_{2},\mathcal{E}_{2},x:A,U\DC V$}
  \BinaryInfC{$\mathcal{R}_{1},\mathcal{R}_{2},\mathcal{E}_{1},\mathcal{E}_{2}, X,U\DC  Y,V$}
  \DisplayProof
\]
Notice that the relation mset of the conclusion may not be treelike even if the premise relation msets have this property. For this reason we use the \textit{additive} version of the cut rule which has the property that the conclusion is an~$\LTSE$ whenever the premises are~$\LTSE$.
\[
  \AxiomC{$\mathcal{R},\mathcal{E}_{1}, X\DC  Y ,x:A$}
  \AxiomC{$\mathcal{R},\mathcal{E}_{2},x:A,U\DC V$}
  \RightLabel{cut}
  \BinaryInfC{$\mathcal{R},\mathcal{E}_{1},\mathcal{E}_{2}, X,U\DC  Y ,V$}
  \DisplayProof
\]
In the above rule, we refer to~$A$ as the \textit{cut-formula}.

The cut-eliminability proof follows that for the labelled calculus~\cite{Neg05}.
\begin{thm}\label{thm_cut_elim}
The cut rule is eliminable in~$\labKst$: LTSE-derivations of the premises of the cut rule can be transformed into an LTSE derivation of the conclusion.
\end{thm}
\begin{proof}
By primary induction on the size of the cut-formula and secondary induction on the sum of the heights of the derivations of the premises.

First suppose that the cut-formula is \textit{not} principal in the left premise (the argument is analogous if the cut-formula is not principal in the right premise). When the last rule~$r$ in the left premise is a unary rule, the cut has the following form (an analogous argument applies to rules with more premises).
\[
  \AxiomC{$\mathcal{R},\mathcal{R}',\mathcal{E}_{1},\mathcal{E}', X_{1}'\DC Y_{1}',x:B$}
  \RightLabel{$r$}
  \UnaryInfC{$\mathcal{R},\mathcal{E}_{1}, X_{1}\DC Y_{1},x:B$}
  \AxiomC{$\mathcal{R},\mathcal{E}_{2},x:B, X_{2}\DC Y_{2}$}
  \RightLabel{cut}
  \BinaryInfC{$\mathcal{R},\mathcal{E}_{1},\mathcal{E}_{2}, X_{1}, X_{2}\DC Y_{1}, Y_{2}$}
  \DisplayProof
\]
If~$r$ has a side condition, since the eigenvariable(s)~$\bar{e}$ do not appear in its conclusion, it follows that no variable in~$\bar{e}$ appears in~$\mathcal{R}$ and hence not in the right premise of cut either (recall by Def.~\ref{def_LTSE}(\ref{enum_def_LTSEi}) that every variable in the LTSE must appear in the relation mset). Proceed as follows.
\[
  \AxiomC{$\mathcal{R},\mathcal{R}',\mathcal{E}_{1},\mathcal{E}', X_{1}'\DC Y_{1}',x:B$}
  \AxiomC{$\mathcal{R},\mathcal{E}_{2},x:B, X_{2}\DC Y_{2}$}
  \RightLabel{hp weakening}
  \dashedLine
  \UnaryInfC{$\mathcal{R},\mathcal{R}',\mathcal{E}_{2},x:B, X_{2}\DC Y_{2}$}
  \RightLabel{cut}
  \BinaryInfC{$\mathcal{R},\mathcal{R}',\mathcal{E}_{1},\mathcal{E}',\mathcal{E}_{2}, X_{1}', X_{2}\DC Y_{1}', Y_{2}$}
  \RightLabel{$r$}
  \UnaryInfC{$\mathcal{R},\mathcal{E}_{1},\mathcal{E}_{2}, X_{1}, X_{2}\DC Y_{1}, Y_{2}$}
  \DisplayProof
\]
Since this cut has identical cut-formula to the original cut and lesser cut-height, it is eliminable by the induction hypothesis.

The remaining case to consider is when the cut-formula is principal in both left and right premise. When the logical rule has no eigenvariables, the transformations are standard. The remaining cases are~($\Box$r) and~($\diam$l). We illustrate the former below (the transformation for the latter is similar).
\[
  \AxiomC{$\mathcal{R},Rxy,Rxz,\mathcal{E}_{1}, X_{1}\DC Y_{1},y:A$}
  \RightLabel{(${\Box}$r)}
  \UnaryInfC{$\mathcal{R},Rxz,\mathcal{E}_{1}, X_{1}\DC Y_{1},x:\Box A$}
  \AxiomC{$\mathcal{R},Rxz,\mathcal{E}_{2},x:\Box A,z:A, X_{2}\DC Y_{2}$}
  \RightLabel{(${\Box}$l)}
  \UnaryInfC{$\mathcal{R},Rxz,\mathcal{E}_{2},x:\Box A, X_{2}\DC Y_{2}$}
  \RightLabel{cut}
  \BinaryInfC{$\mathcal{R},Rxz,\mathcal{E}_{1},\mathcal{E}_{2}, X_{1}, X_{2}\DC Y_{1}, Y_{2}$}
  \DisplayProof
\]
Then proceed:
\begin{small}
  \[
    \AxiomC{$\mathcal{R},Rxy,Rxz,\mathcal{E}_{1}, X_{1}\DC Y_{1},y:A$}
    \RightLabel{LTSE-subs.}
    \dashedLine
    \UnaryInfC{$\mathcal{R},Rxz,Rxz,\mathcal{E}_{1}, X_{1}\DC Y_{1},z:A$}
    \RightLabel{ctr}
    \dashedLine
    \UnaryInfC{$\mathcal{R},Rxz,\mathcal{E}_{1}, X_{1}\DC Y_{1},z:A$}
    \AxiomC{$\mathcal{R},Rxz,\mathcal{E}_{1}, X_{1}\DC Y_{1},x:\Box A$}
    \noLine
    \LeftLabel{cut$\Big\{$}
    \UnaryInfC{$\mathcal{R},Rxz,\mathcal{E}_{2},x:\Box A,z:A, X_{2}\DC Y_{2}$}
    \UnaryInfC{$\mathcal{R},Rxz,\mathcal{E}_{1},\mathcal{E}_{2},z:A, X_{1}, X_{2}\DC Y_{1}, Y_{2}$}
    \RightLabel{cut}
    \BinaryInfC{$\mathcal{R},Rxz,\mathcal{E}_{1},\mathcal{E}_{1},\mathcal{E}_{2}, X_{1}, X_{1}, X_{2}\DC Y_{1}, Y_{1}, Y_{2}$}
    \dashedLine
    \RightLabel{ctr}
    \UnaryInfC{$\mathcal{R},Rxz,\mathcal{E}_{1},\mathcal{E}_{2}, X_{1}, X_{2}\DC Y_{1}, Y_{2}$}
    \DisplayProof
  \]
\end{small}
The upper cut has identical cut-formula to the original cut and lesser cut-height so it is eliminable by the induction hypothesis. 
Use the cutfree derivation as the right premise of the lower cut in the above diagram. This cut has smaller cut-formula so it is eliminable by the induction hypothesis. In this way we obtain a cutfree derivation.
\end{proof}

\section{Geach logics: soundness and completeness for LTSE-derivations}\label{sec_Geach_derivations}

We are ready to show that LTSE-derivations in~$\labKst$ are sound and complete for the corresponding Geach logic.
This means that $\DC x:A$ is derivable~$\labKst$ iff $A$~is a theorem of the corresponding Geach logic.

First of all, here is the LTSE derivation of~$G(h,i,j,k)$ using $\L(h,i,j,k)$:
\begin{small}
  \[
    \AxiomC{$\RR^{h}xy,\RR^{j}xz,\RR^{i}yu,\RR^{k}zu,u=u',y:\Box^{i} p,\ldots,u':p, u:p\DC z:\diam^{k} p,\ldots,u':p$}
    \RightLabel{(rep-l) on $u=u'$}
    \UnaryInfC{$\RR^{h}xy,\RR^{j}xz,\RR^{i}yu,\RR^{k}zu,u=u',y:\Box^{i} p,\ldots,u:p\DC z:\diam^{k} p,\ldots,u':p$}
    \RightLabel{($\diam$r) $k$ times}
    \UnaryInfC{$\RR^{h}xy,\RR^{j}xz,\RR^{i}yu,\RR^{k}zu,u=u',y:\Box^{i} p,\ldots,u:p\DC z:\diam^{k} p$}
    \RightLabel{($\Box$l) $i$ times}
    \UnaryInfC{$\RR^{h}xy,\RR^{j}xz,\RR^{i}yu,\RR^{k}zu',u=u',y:\Box^{i} p\DC z:\diam^{k} p$}
    \RightLabel{$\L(h,i,j,k)$}
    \UnaryInfC{$\RR^{h}xy,\RR^{j}xz,y:\Box^{i} p\DC z:\diam^{k} p$}
    \RightLabel{($\Box$r)  $j$ times}
    \UnaryInfC{$\RR^{h}xy, y:\Box^{i} p\DC x:\Box^{j}\diam^{k} p$}
    \RightLabel{(${\diam}$l) $h$ times}
    \UnaryInfC{$x:\diam^{h}\Box^{i} p \DC x:\Box^{j}\diam^{k} p$}
    \RightLabel{(${\imp}$r)}
    \UnaryInfC{$\DC x:\diam^{h}\Box^{i} p\imp \Box^{j}\diam^{k} p$}
    \DisplayProof
  \]
\end{small}

\begin{thm}\label{thm_s_and_c}
Let~$S\subset\mathbb{N}^{4}$ be finite.
$A\in K+\{G(h,i,j,k) | (h,i,j,k)\in S\}$ iff $\DC x:A$ is LTSE-derivable in $\labK+\{\L(h,i,j,k) | (h,i,j,k)\in S\}$.
\end{thm}
\begin{proof}
($\Leftarrow$) is soundness. In the spirit of reusing as much as possible, let us exploit the LS (without equality!) calculus (call it~$\mathcal{C}$) for $K+\{G(h,i,j,k) | (h,i,j,k)\in S\}$ presented in~\cite{Neg05}.  Interpret the LTSE $s=\mathcal{R},\mathcal{E},X\DC Y$ as the LS (without equality) $\mathcal{I}(s)=\mathcal{R}[\mathcal{E}], X[\mathcal{E}]\DC Y[\mathcal{E}]$. Here $(\cdot)[\mathcal{E}]$ denotes the substitution of every $x_{j}$ with~$x_{i}$ where~$i$ is the least number such that $\mathcal{E}\models x_{i}=x_{j}$ ($\models$ is the usual consequence relation for the theory of equality). By inspection, for every rule (below left) in $\labK+\{\L(h,i,j,k) | (h,i,j,k)\in S\}$---taking care first to rename the eigenvariables so that they have higher variable index than non-eigenvariables---the rule below right is admissible in~$\mathcal{C}$.
\begin{center}
\begin{tabular}{c@{\hspace{3cm}}c}
\AxiomC{$s_{1}$}
\AxiomC{$\ldots$}
\AxiomC{$s_{N}$}
\TrinaryInfC{$s_{0}$}
\DisplayProof
&
\AxiomC{$\mathcal{I}(s_{1})$}
\AxiomC{$\ldots$}
\AxiomC{$\mathcal{I}(s_{N})$}
\TrinaryInfC{$\mathcal{I}(s_{0})$}
\DisplayProof
\end{tabular}
\end{center}
It follows then that if $\DC x:A$ is LTSE-derivable in $\labK+\{\L(h,i,j,k) | (h,i,j,k)\in S\}$, then $\DC x:A$ is derivable in~$\mathcal{C}$. By soundness of the latter, $A\in K+\{G(h,i,j,k) | (h,i,j,k)\in S\}$.

($\Rightarrow$) is completeness.
We saw above that~$G(h,i,j,k) | (h,i,j,k)\in S\}$ is $\LTSE$-derivable.
The LTSE-derivation of the normality axiom is given in~(\ref{lab-der-normal-axiom}).
The necessitation rule $\DC x:A / \DC x:\Box A$ is clearly LTSE-derivable. Simulate \textit{modus ponens} using the cut rule.
So from a Hilbert calculus proof of $A\in K+\{G(h,i,j,k) | (h,i,j,k)\in S\}$ we can obtain a LTSE-derivation of $\DC x:A$ in $\labK+\{\L(h,i,j,k) | (h,i,j,k)\in S\} + cut$. The result follows from the cut-elimination theorem (Thm.~\ref{thm_cut_elim}).
\end{proof}

\section{Maps between~$\labKst$ and~$\INSKst$ calculi for Geach logics}\label{sec_relationship}

Throughout this section, let~$S\subset\mathbb{N}^{4}$ be finite.
We now establish a correspondence between LTSE derivations in~$\labK+\{\L(h,i,j,k) | (h,i,j,k)\in S\}$ and $\INSK+\{\I(h,i,j,k) | (h,i,j,k)\in S\}$.
This result extends the correspondence between NS and LTS derivations in~\cite{GorRam12AIML}.
We have not yet introduced the INS Geach structural rule~$\I(h,i,j,k)$ yet (this is done at the end of this section), since the rule can be obtained directly from~$\L(h,i,j,k)$ using the maps below.

\subsection{From INS to LTSE}

\begin{description}
\item[Input]~INS~$s$

\item[Step 1]~Represent~$s$ as its underlying tree~$\tau_{s}$ whose nodes are decorated by indexed nested sequents. 

\item[Step 2]~Name the nodes of~$\tau_{s}$ as~$x_{1},\ldots,x_{N}$ and let $X_{k}\DC^{a_{k}}Y_{k}$ be the decoration at~$x_{k}$.
Let~$S$ denote the set $\cup_{i=1}^{N} \{a_{i}\}$ of all indices in~$s$. Describe the structure of the tree~$\tau_{s}$ as a relation mset~$\mathcal{R}$.

\item[Step 3]~For each~$a\in S$, define~$\mathcal{E}_{a}=\left\{ x_{i}=x_{j} \,|\, \text{$i\neq j$ and $a_{i}=a_{j}$} \right\}$ and set~$\mathcal{E}=\cup_{a\in S}\mathcal{E}_{a}$.

\item[Output the LTSE]
\[
\mathcal{R},\mathcal{E},\cup_{i=1}^{N}\{x_{i}:B| B\in X_{i}\} \DC \cup_{i=1}^{N}\{x_{i}:B| B\in Y_{i}\}
\]
\end{description}

\begin{exa}
Suppose that the input INS~$s$ is
\[
X\DC^{0} Y, [P\DC^{1} Q, [U\DC^{0} V], [L\DC^{1} M]], [S\DC^{2} T]
\]
Step~1 represents~$s$ as the tree~$\tau_{s}$ (below left)
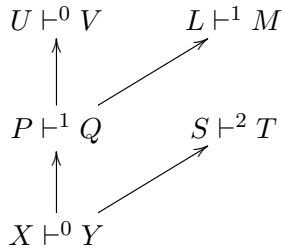
\begin{align*}
\text{
\xymatrix{
U\DC^{0} V & L\DC^{1} M \\
P\DC^{1} Q \ar[u]\ar[ur] & S\DC^{2} T \\
X\DC^{0} Y\ar[u]\ar[ur] & 
}
}
&
\quad
&
\text
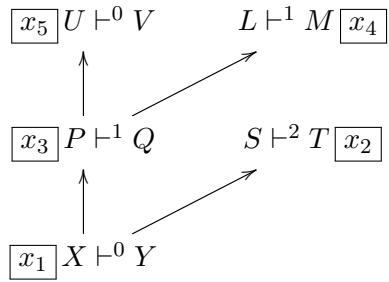
{
\xymatrix{
\fbox{$x_{5}$}\,U\DC^{0} V & L\DC^{1} M \, \fbox{$x_{4}$} \\
\fbox{$x_{3}$}\,P\DC^{1} Q \ar[u]\ar[ur] & S\DC^{2} T \, \fbox{$x_{2}$}\\
\fbox{$x_{1}$}\, X\DC^{0} Y\ar[u]\ar[ur] & 
}
}
\end{align*}
In Step~2 we label each node of~$\tau_{s}$ with variables $x_{1},\ldots,x_{5}$ as indicated above right. The set~$S$ of indices is $\{0,1,2\}$. The relation mset describing~$\tau_{s}$ is $\mathcal{R}=\{Rx_{1}x_{2},Rx_{1}x_{3},Rx_{3}x_{4},Rx_{3}x_{5}\}$.

For Step~3, observe that in the tree above right, $x_{1}$ and~$x_{5}$ both have the index~$0$. Meanwhile $x_{3}$ and~$x_{4}$ both have the index~$1$. Only~$x_{2}$ has the index~$2$. Thus $\mathcal{E}_{0}=\{x_{1}=x_{5}\}$, 
$\mathcal{E}_{1}=\{x_{3}=x_{4}\}$ and $\mathcal{E}_{2}=\emptyset$. Taking their union, $\mathcal{E}=\{x_{1}=x_{5}, x_{3}=x_{4}\}$.

For the output, first read off $L=\{x_{1}:X, x_{2}:S, x_{3}:P, x_{4}:L, x_{5}:U\}$ and $M=\{x_{1}:Y, x_{2}:T, x_{3}:Q, x_{4}:M, x_{5}:V\}$ from the tree above right. Output the LTSE
\[
Rx_{1}x_{2},Rx_{1}x_{3},Rx_{3}x_{4},Rx_{3}x_{5}, x_{1}=x_{5}, x_{3}=x_{4}, L\DC M
\]
\end{exa}

\subsection{From LTSE to INS} 

\begin{description}
\item[Input] LTSE~$\REL,X\DC Y$

\item[Step 1] Construct the tree~$\tau_{s}$ defined by~$\mathcal{R}$.

\item[Step 2] Decorate each node~$u$ in~$\tau_{s}$ with the traditional sequent $X_{u}\DC Y_{u}$. Here $X_{u}$ ($Y_{u}$) denotes the subset of~$X$ (resp.~$Y$) of formulae with label~$u$.

\item[Step 3] It remains to assign an index to the traditional sequents. Assign the index so that two traditional sequents have the same index iff their corresponding nodes $u,v\in\mathbb{S}$ in~$\tau_{s}$ have the property that $\mathcal{E}\models u=v$.

\item[Output] Represent the decorated tree as an INS.
\end{description}

\begin{exa}
Set $X_{0}=\{x_{1}:X, x_{2}:S, x_{3}:P, x_{4}:L, x_{5}:U\}$ and $Y_{0}=\{x_{1}:Y, x_{2}:T, x_{3}:Q, x_{4}:M, x_{5}:V\}$.
Suppose that we are given the LTSE~$s$:
\[
Rx_{1}x_{2},Rx_{1}x_{3},Rx_{3}x_{4},Rx_{3}x_{5}, x_{1}=x_{5}, x_{3}=x_{4}, X_{0}\DC Y_{0}
\]
From the relation mset we obtain the tree~$\tau_{s}$ below left.
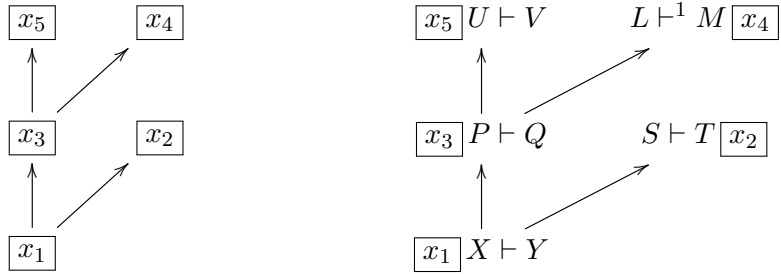
\begin{align*}
\text{
\xymatrix{
\fbox{$x_{5}$} & \fbox{$x_{4}$} \\
\fbox{$x_{3}$} \ar[u]\ar[ur] & \fbox{$x_{2}$}\\
\fbox{$x_{1}$} \ar[u]\ar[ur] & 
}
}
&
\quad
&
\text{
\xymatrix{
\fbox{$x_{5}$}\,U\DC V & L\DC^{1} M \, \fbox{$x_{4}$} \\
\fbox{$x_{3}$}\,P\DC Q \ar[u]\ar[ur] & S\DC T \, \fbox{$x_{2}$}\\
\fbox{$x_{1}$}\, X\DC Y\ar[u]\ar[ur] & 
}
}
\end{align*}
In Step~2, we decorate this tree with the labelled formulae to obtain the decorated tree above right.

The equality mset of~$s$ is $\{x_{1}=x_{5}, x_{3}=x_{4}\}$. Following Step~3, we will assign the same index to~$x_{1}$ and~$x_{5}$, and also the same index to~$x_{3}$ and~$x_{4}$. The remaining nodes are assigned further distinct indices. In this way we obtain the decorated tree below.
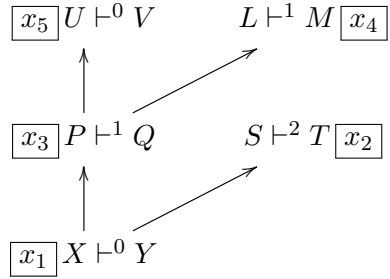
\begin{center}
\text{
\xymatrix{
\fbox{$x_{5}$}\,U\DC^{0} V & L\DC^{1} M \, \fbox{$x_{4}$} \\
\fbox{$x_{3}$}\,P\DC^{1} Q \ar[u]\ar[ur] & S\DC^{2} T \, \fbox{$x_{2}$}\\
\fbox{$x_{1}$}\, X\DC^{0} Y\ar[u]\ar[ur] & 
}
}
\end{center}
From this tree we read off the corresponding INS
\[
X\DC^{0} Y, [P\DC^{1} Q, [U\DC^{0} V], [L\DC^{1} M]], [S\DC^{2} T]
\]
\end{exa}

\subsection{Between LTSE-derivations in~$\labKst$ and derivations in~$\INSKst$.}\label{translation}

Suppose that~$r$ is one of the following rules in~$\labK$:
\[
r\in \{\text{init}, \text{init}{-}\bot, {\lor}l, {\lor }r, {\land}l, {\land}r, {\imp}l,{\imp}r, \Box l, \Box r\}
\]
By inspection, for every rule instance, applying the LTSE-to-INS map above to the premise(s) and conclusion yields a rule instance of~($r$) in~$\INSK$.

Similarly, suppose that~$r$ is one of the following rules in~$\INSK$:
\[
r\in \{\text{init}, \text{init}{-}\bot, {\lor}l, {\lor }r, {\land}l, {\land}r, {\imp}l,{\imp}r, \Box l, \Box r\}
\]
By inspection, for every rule instance, applying the INS-to-LTSE map above to the premise(s) and conclusion yields a rule instance of~($r$) in~$\labK$. 

Similarly, the rule (fc-l) corresponds to the rule (rep-l), and the rule (fc-r) corresponds to the rule (rep-r).
Also, the (sc) rule corresponds to the rule (ls-sc). 

Finally, for a LTSE Geach structural rule ($\L(h,i,j,k)$)
\[
  \AxiomC{$\mathcal{R},\RR^{h}xy,\RR^{j}xz,\RR^{i}\lambda(x\bar{y})u,\RR^{k}\lambda(x\bar{z})v,\lambda(x\bar{y}\bar{u})=\lambda(x\bar{z}\bar{v}),\mathcal{E},U\DC V$}
  \RightLabel{$\L(h,i,j,k)$}
  \UnaryInfC{$\mathcal{R},\RR^{h}xy,\RR^{j}xz,\mathcal{E},U\DC V$}
  \DisplayProof
\]
\noindent define the following \textit{INS Geach structural rule} $\I(h,i,j,k)$) using the LTSE-to-INS map. 

\textbf{First suppose that $h,j>0$.} Then~$\I(h,i,j,k)$ has conclusion
\begin{multline*}
\Gamma[X_{1} \DC^{a_{1}} X_{1}', [X_{2}\DC^{a_{2}} X_{2}',[\ldots[X_{h-1}\DC^{a_{h-1}} X_{h-1}', [X_{h}\DC^{a_{h}} X_{h}]]\ldots]], \\
[Y_{1} \DC^{b_{1}} Y_{1}', [Y_{2}\DC^{b_{2}} Y_{2}',[\ldots[Y_{n-1}\DC^{b_{j-1}} Y_{j-1}', [Y_{j}\DC^{b_{j}} Y_{j}']]\ldots]]] ]
\end{multline*}
\begin{itemize}
\item If~$i,j>0$ then~$\I(h,i,j,k)$ has the following premise where the $c_{1},\ldots,d_{1}\ldots, c$ do not appear in the conclusion.
\begin{multline*}
\Gamma[ [X_{1} \DC^{a_{1}} X_{1}', [X_{2}\DC^{a_{2}} X_{2}',[\ldots[X_{h-1}\DC^{a_{h-1}} X_{h-1}', [X_{h}\DC^{a_{h}} X_{h}, \\
	[\, \DC^{c_{1}} \,, [\, \DC^{c_{2}} \,,[\ldots[\, \DC^{c_{i-1}} \,, [\, \DC^{c} \,]]\ldots]]] ]]\ldots]]] , \\
[Y_{1} \DC^{b_{1}} Y_{1}', [Y_{2}\DC^{b_{2}} Y_{2}',[\ldots[Y_{n-1}\DC^{b_{j-1}} Y_{j-1}', [Y_{j}\DC^{b_{j}} Y_{j}', \\
	[\, \DC^{d_{1}} \,, [\, \DC^{d_{2}} \,,[\ldots[\, \DC^{d_{k-1}} \,, [\, \DC^{c} \,]]\ldots]]] ]]\ldots]]] ]
\end{multline*}
\item If~$i=0$ and~$j>0$ then~$\I(h,i,j,k)$ has the following premise where the $d_{1}\ldots, c$ do not appear in the conclusion.
\begin{multline*}
\Gamma[ [X_{1} \DC^{a_{1}} X_{1}', [X_{2}\DC^{a_{2}} X_{2}',[\ldots[X_{h-1}\DC^{a_{h-1}} X_{h-1}', [X_{h}\DC^{c} X_{h}, \\
[Y_{1} \DC^{b_{1}} Y_{1}', [Y_{2}\DC^{b_{2}} Y_{2}',[\ldots[Y_{n-1}\DC^{b_{j-1}} Y_{j-1}', [Y_{j}\DC^{b_{j}} Y_{j}', \\
	[\, \DC^{d_{1}} \,, [\, \DC^{d_{2}} \,,[\ldots[\, \DC^{d_{k-1}} \,, [\, \DC^{c} \,]]\ldots]]] ]]\ldots]]] ]
\end{multline*}
\item If~$i,j=0$ then~$\I(h,i,j,k)$ has the following premise (there are no eigenvariables)
\begin{multline*}
\Gamma[ [X_{1} \DC^{a_{1}} X_{1}', [X_{2}\DC^{a_{2}} X_{2}',[\ldots[X_{h-1}\DC^{a_{h-1}} X_{h-1}', [X_{h}\DC^{a_{h}} X_{h}]]\ldots]]], \\
[Y_{1} \DC^{b_{1}} Y_{1}', [Y_{2}\DC^{b_{2}} Y_{2}',[\ldots[Y_{n-1}\DC^{b_{j-1}} Y_{j-1}', [Y_{j}\DC^{a_{h}} Y_{j}']]\ldots]]] ]\\
\end{multline*}
\end{itemize}

\textbf{Next, suppose that~$h=0$ and~$j>0$.} Then~$\I(h,i,j,k)$ has conclusion
\[
\Gamma[U\DC^{n} V, [Y_{1} \DC^{b_{1}} Y_{1}', [Y_{2}\DC^{b_{2}} Y_{2}',[\ldots[Y_{n-1}\DC^{b_{j-1}} Y_{j-1}', [Y_{j}\DC^{b_{j}} Y_{j}']]\ldots]]] ]
\]
\begin{itemize}
\item If~$i,j>0$ then~$\I(h,i,j,k)$ has the following premise where the $c_{1},\ldots,d_{1}\ldots, c$ do not appear in the conclusion.
\begin{multline*}
\Gamma[ U\DC^{n} V, [\, \DC^{c_{1}} \,, [\, \DC^{c_{2}} \,,[\ldots[\, \DC^{c_{i-1}} \,, [\, \DC^{c} \,]]\ldots]]] , \\
[Y_{1} \DC^{b_{1}} Y_{1}', [Y_{2}\DC^{b_{2}} Y_{2}',[\ldots[Y_{n-1}\DC^{b_{j-1}} Y_{j-1}', [Y_{j}\DC^{b_{j}} Y_{j}', \\
	[\, \DC^{d_{1}} \,, [\, \DC^{d_{2}} \,,[\ldots[\, \DC^{d_{k-1}} \,, [\, \DC^{c} \,]]\ldots]]] ]]\ldots]]] ]
\end{multline*}
\item If~$i=0$ and~$j>0$ then~$\I(h,i,j,k)$ has the following premise where the $d_{1}\ldots, c$ do not appear in the conclusion.
\begin{multline*}
\Gamma[ U\DC^{c} V, [Y_{1} \DC^{b_{1}} Y_{1}', [Y_{2}\DC^{b_{2}} Y_{2}',[\ldots[Y_{n-1}\DC^{b_{j-1}} Y_{j-1}', [Y_{j}\DC^{b_{j}} Y_{j}', \\
	[\, \DC^{d_{1}} \,, [\, \DC^{d_{2}} \,,[\ldots[\, \DC^{d_{k-1}} \,, [\, \DC^{c} \,]]\ldots]]] ]]\ldots]]] ]
\end{multline*}
\item If~$i,j=0$ then~$\I(h,i,j,k)$ has the following premise (there are no eigenvariables)
\[
\Gamma[ U\DC^{n} V, [Y_{1} \DC^{b_{1}} Y_{1}', [Y_{2}\DC^{b_{2}} Y_{2}',[\ldots[Y_{n-1}\DC^{b_{j-1}} Y_{j-1}', [Y_{j}\DC^{n} Y_{j}']]\ldots]]] ]
\]
\end{itemize}
The remaining cases are analogous. Of course, if $h,i,j,k=0$ this corresponds to~$p\imp p$ so no structural rule is required.

From Thm.~\ref{thm_s_and_c} and the above translations we get:
\begin{cor}\label{cor-INSK}
Let~$S\subset\mathbb{N}^{4}$ be finite. Then $\INSK+\{\I(h,i,j,k) | (h,i,j,k)\in S\}$ is a INS calculus for $K+\{G(h,i,j,k) | (h,i,j,k)\in S\}$. Moreover, there is a 1:1 correspondence between a derivation of $\DC A$ in this calculus and a LTSE-derivation of $\DC x:A$ in $\labK+\{\L(h,i,j,k) | (h,i,j,k)\in S\}$.
\end{cor}
This Corollary covers \textit{all} the Geach logics, including the axioms $\{\diam^{h}p\imp \Box^{j}p|h,j\geq 0\}$ which are omitted in~\cite{Fit15}.
\begin{exa}
Here are the INS and LTSE rules for $\diam p\imp\Box p$.
\medskip
\begin{center}
\begin{tabular}{c@{\hspace{1cm}}c}
\AxiomC{$\Gamma\left[[X\DC^{a} Y\bl\context],[U\DC^{a} V\bl\contexta]\right]$}
\RightLabel{$\I(1,0,1,0)$}
\UnaryInfC{$\Gamma\left[[X\DC^{a} Y\bl\context],[U\DC^{b} V\bl\contexta]\right]$}
\DisplayProof
&
\AxiomC{$\REL[,Rxy,Rxz],y=z,X\DC Y$}
\RightLabel{$\L(1,0,1,0)$}
\UnaryInfC{$\REL[,Rxy,Rxz],X\DC Y$}
\DisplayProof
\end{tabular}
\end{center}
\medskip
The INS derivation of the axiom and the corresponding LTSE derivation are given below left and right respectively.
\medskip
\begin{center}
\begin{tabular}{cc}
\AxiomC{$\, \DC^{0} \, , [p \DC^{1} \,], [p \DC^{1} p]$}
\RightLabel{(fc-l)}
\UnaryInfC{$\, \DC^{0} \, , [p \DC^{1}\,], [\, \DC^{1} p]$}
\RightLabel{$\I(1,0,1,0)$}
\UnaryInfC{$\, \DC^{0} \, , [p \DC^{1}\,], [\, \DC^{2} p]$}
\RightLabel{($\Box$l)}
\UnaryInfC{$\, \DC^{0} \Box p, [p \DC^{1}\,]$}
\RightLabel{($\diam$l)}
\UnaryInfC{$\diam p\DC^{0} \Box p$}
\RightLabel{(${\imp}$r)}
\UnaryInfC{$\DC^{0} \diam p\imp \Box p$}
\DisplayProof
&
\AxiomC{$Rxy, Rxz, y=z, y:p, z:p \DC z: p$}
\RightLabel{(rep-l)}
\UnaryInfC{$Rxy, Rxz, y=z, y:p \DC z: p$}
\RightLabel{$\L(1,0,1,0)$}
\UnaryInfC{$Rxy, Rxz, y:p \DC z: p$}
\RightLabel{($\Box$l)}
\UnaryInfC{$Rxy, y:p \DC x:\Box p$}
\RightLabel{($\diam$l)}
\UnaryInfC{$x:\diam p\DC x:\Box p$}
\RightLabel{(${\imp}$r)}
\UnaryInfC{$\DC x:\diam p\imp \Box p$}
\DisplayProof
\end{tabular}
\end{center}
\medskip
\end{exa}

\begin{exa}\label{exa_INS_from_Geach}
Consider~$G(2,2,1,0)$. Then the LTSE-rule~$\L(2,2,1,0)$ is given below, where $u,u_{1}$ do not appear in the conclusion.
\[
  \AxiomC{$\mathcal{R},Rxy_{1},Ry_{1}y,Rxz,Ryu_{1},Ru_{1}u,u=z,\mathcal{E}, X\DC Y$}
  \RightLabel{$\L(2,2,1,0)$}
  \UnaryInfC{$\mathcal{R},Rxy_{1},Ry_{1}y,Rxz,\mathcal{E}, X\DC Y$}
  \DisplayProof
\]
The corresponding INS rule $\I(2,2,1,0)$ is
\[
  \AxiomC{$\Gamma\left[ [X_{1}\DC^{a_{1}} Y_{1}, [X_{2}\DC^{a_{2}} Y_{2}, [\, \DC^{c_{1}} \, ,[\, \DC^{c} \,]]]], [L_{1}\DC^{c} M_{1}]\right]$}
  \RightLabel{$\I(2,2,1,0)$}
  \UnaryInfC{$\Gamma\left[ [X_{1}\DC^{a_{1}} Y_{1}, [X_{2}\DC^{a_{2}} Y_{2}]], [L_{1}\DC^{b_{1}} M_{1}]\right]$}
  \DisplayProof
\]
\end{exa}

\section{Intermediate logics}

%
Logics between classical and intuitionistic logic are called \textit{intermediate logics}.
While nested sequent calculi have been presented for intuitionistic logic~Ip~\cite{Fit14} and recently also for extensions of intuitionistic logic with modalities~\cite{Str13FOSSACS,MarStr14AIML}, we are not aware of any nested sequent calculi for intermediate logics. 
In contrast, a LS calculus~$\labI$~\cite{DycNeg12} (Fig.~\ref{fig_labI}) for propositional intuitionistic logic was extended via structural rules to capture those intermediate logics whose Kripke semantics are defined by geometric formulae. 

\begin{figure}

\begin{small}
\begin{tabular}{cc}
\AxiomC{}
\RightLabel{(init-$\bot$)}
\UnaryInfC{$\REL,x:\bot,X\DC  Y $}
\DisplayProof
&
\AxiomC{}
\UnaryInfC{$\REL[,Rxy],x:p, X\DC  Y, y:p$}
\DisplayProof
\\[1em]
\AxiomC{$\REL[,Rxv,v=x], X\DC Y$}
\RightLabel{(iref)}
\UnaryInfC{$\REL, X\DC Y$}
\DisplayProof
&
\AxiomC{$\REL[,Rxy,Ryz,Rxv,v=z], X\DC Y$}
\RightLabel{(itrans)}
\UnaryInfC{$\REL[,Rxy,Ryz], X\DC Y$}
\DisplayProof
\\[1em]
 \AxiomC{$\REL,x:A,x:B, X\DC  Y$}
 \RightLabel{${\land}$l}
 \UnaryInfC{$\REL,x:A\land B, X\DC  Y$}
 \DisplayProof
 &
 \AxiomC{$\REL, X\DC  Y,x:A$}
 \AxiomC{$\REL, X\DC  Y,x:B$}
 \RightLabel{${\land}$r}
 \BinaryInfC{$\REL, X\DC  Y,x:A\land B$}
 \DisplayProof
\\[1em]
 \AxiomC{$\REL,x:A, X\DC Y$}
 \AxiomC{$\REL,x:B, X\DC Y$}
 \RightLabel{(${\lor}$l)}
 \BinaryInfC{$\REL,x:A\lor B, X\DC Y$}
 \DisplayProof
 &
 \AxiomC{$\REL, X\DC Y,x:A,x:B$}
 \RightLabel{(${\lor}$r)}
 \UnaryInfC{$\REL, X\DC Y,x:A\lor B$}
 \DisplayProof
\\[1.5em]
\multicolumn{2}{c}{
\AxiomC{$\REL[,Rxy], x:A\imp B, X\DC y:A, Y$}
\AxiomC{$\REL[,Rxy], x:A\imp B,y:B, X\DC  Y$}
\RightLabel{(${\imp}$l)}
\BinaryInfC{$\REL[,Rxy], x:A\imp B,y:B, X\DC  Y$}
\DisplayProof
}
\\[1.5em]
\AxiomC{$\REL[,Rxv],v:A,  X\DC Y,v:B$}
\RightLabel{(${\imp}$r)}
\UnaryInfC{$\REL, X\DC Y,x:A\imp B$}
\DisplayProof
&
\AxiomC{$\REL[,Rxy,Ruv],x=u,y=v,X\DC Y $}
\RightLabel{(ls-sc)}
\UnaryInfC{$\REL[,Rxy],x=u,X\DC Y $}
\DisplayProof
\\[1em]
\AxiomC{$\REL,x=y,x:A,y:A,X\DC Y $}
\RightLabel{(rep-l)}
\UnaryInfC{$\REL,x=y,x:A,X\DC Y $}
\DisplayProof
&
\AxiomC{$\REL,x=y,X\DC Y ,x:A,y:A$}
\RightLabel{(rep-r)}
\UnaryInfC{$\REL,x=y,X\DC Y ,x:A$}
\DisplayProof
\\[1em]
\AxiomC{$\REL, x=x,X\DC Y $}
\RightLabel{(eq-ref)}
\UnaryInfC{$\REL,X\DC Y $}
\DisplayProof
&
\AxiomC{$\REL, x=z, x=y, y=z, X\DC Y $}
\RightLabel{(eq-trans)}
\UnaryInfC{$\REL, x=y, y=z,X\DC Y $}
\DisplayProof
\end{tabular}
\end{small}

\caption{The labelled sequent calculus~$\labI$. In (${\imp}$r), (iref), (itrans) and (ls-sc), $v$ does not appear in the conclusion.}
\label{fig_labI}
\end{figure}

In this section we show how to use~$\labI$ to obtain an INS calculi for intermediate logics whose frame semantics are defined by first-order formulae of 
the form~(\ref{Geach-fo}). From the corresponding~$(h,i,j,k)$ we obtain labelled sequent rule~$\L(h,i,j,k)$ and INS rule~$\I(h,i,j,k)$.
The INS calculus~$\INSI$ (Fig.~\ref{fig_INSI}) is obtained from the rules of~$\labI$ by translating each~$\LTSE$ into an~INS using the LTSE-to-INS map from the previous section.
By a similar argument to that in Sec.~\ref{sec_results} we can show that the cut rule is eliminable from~$\labI$ while preserving $\LTSE$-derivations.
In this way we can obtain a 1:1 correspondence between LTSE-derivations in~$\labI+\{ \L(h,i,j,k) | (h,i,j,k)\in S \}$ and derivations in $\INSI + \{ \I(h,i,j,k) | (h,i,j,k) \in S \}$ for finite $S\subset\mathbb{N}^{4}$. 

The well-known relationship (see e.g.~\cite{ChaZak97}) between intermediate logics and their frame semantics yields, for example, that $Ip+(p\imp\bot)\lor ((p\imp\bot)\imp\bot)$ is the logic of frames satisfying
\[
\forall xyz(Rxy\land Rxz\imp \exists uv(Ryu\land Rzv\land u=v))\quad\quad\quad\quad\quad
\]
We saw in Example~\ref{eg-conf} that the corresponding Geach structural rule is~(\ref{LTSE_Geach_rule}). Hence the INS calculus for this logic is~$\INSI+\I(1,1,1,1)$ (see~(\ref{INS_Geach_rule}).
Here is the LTSE-derivation of the axiom.
\[
  \AxiomC{$Rxy,Rxz,Ryu,Rzv,y=v,y:p,{v:p},z:p\imp\bot\DC y:\bot, z:\bot, {v:p}$}
  \RightLabel{(rep-l)}
  \UnaryInfC{$Rxy,Rxz,Ryu,Rzv,y=v,y:p,z:p\imp\bot\DC y:\bot, z:\bot, v:p$}
  \RightLabel{\vspace{1cm}\hspace{-1.5cm}(${\imp}$l)$\Big\{$}
  \noLine
  \UnaryInfC{\vspace{0.1cm}}
  \noLine
  \UnaryInfC{$Rxy, Rxz,Ryu,Rzv,u=v, y:p, z:(p\imp\bot), v:\bot \DC x:\bot, z:\bot$}
  \UnaryInfC{$Rxy, Rxz,Ryu,Rzv,u=v, y:p, z:(p\imp\bot) \DC x:\bot, z:\bot$}
  \RightLabel{$\L(1,1,1,1)$}
  \UnaryInfC{$Rxy, Rxz, y:p, z:(p\imp\bot) \DC x:\bot, z:\bot$}
  \RightLabel{(${\imp}$r)}
  \UnaryInfC{$Rxy, y:p \DC y:\bot, ((p\imp\bot)\imp\bot)$}
  \RightLabel{(${\imp}$r)}
  \UnaryInfC{$\DC x:(p\imp\bot), ((p\imp\bot)\imp\bot)$}
  \RightLabel{(${\lor}$r)}
  \UnaryInfC{$\DC x:(p\imp\bot)\lor ((p\imp\bot)\imp\bot)$}
  \DisplayProof
\]
We leave it to the reader to construct the corresponding INS derivation.

\begin{figure}[t]
\begin{center}
\begin{small}
\begin{tabular}{cc}
\AxiomC{}
\RightLabel{(init)}
\UnaryInfC{$\Gamma\ob p,X\DC^{a} Y,p\bl\context\cb$}
\DisplayProof
&
\AxiomC{}
\UnaryInfC{$\Gamma\ob p,X\DC^{a} Y, [U\DC^{b} V,p\bl\context]\cb$}
\DisplayProof
\\[1em]
\AxiomC{$\Gamma\ob X\DC^{a} Y,[\,\DC^{a}]\bl\context\cb$}
\RightLabel{(iref)}
\UnaryInfC{$\Gamma\ob X\DC^{a} Y\bl\context\cb$}
\DisplayProof
&
\AxiomC{$\Gamma\ob X\DC^{a} Y\bl\context,[U\DC^{b} V\bl\contexta,[P\DC^{s} Q\bl\contextb]], [\,\DC^{s}\,]\cb$}
\RightLabel{(itrans)}
\UnaryInfC{$\Gamma\ob X\DC^{a} Y\bl\context,[U\DC^{b} V\bl\contexta,[P\DC^{s} Q\bl\contextb]],\cb$}
\DisplayProof
\\[1.5em]
\multicolumn{2}{c}{
\AxiomC{$\Gamma\ob A\imp B,X\DC^{a} Y\bl\context, [U\DC^{b} V,A\bl\contexta]\cb$}
\AxiomC{$\Gamma\ob A\imp B,X\DC^{a} Y\bl\context, [B,U\DC^{b} V\bl\contexta]\cb$}
\RightLabel{(${\imp}$l)}
\BinaryInfC{$\Gamma\ob A\imp B,X\DC^{a} Y\bl\context, [U\DC^{b} V\bl\contexta]\cb$}
\DisplayProof
}
\\[1.5em]
\AxiomC{$\Gamma\ob X\DC^{a} Y,[A\DC^{b} B]\bl\context\cb$}
\RightLabel{(${\imp}$r)}
\UnaryInfC{$\Gamma\ob X\DC^{a} Y,A\imp B\bl\context\cb$}
\DisplayProof
&
\AxiomC{$\Gamma\left\ob X\DC^{a} Y\bl\context, [P\DC^{b} Q\bl\contextb]\right\cb\left\ob U\DC^{a} V\bl\contexta, [\phantom{a}\DC^{b}\phantom{a}] \right\cb$}
\RightLabel{(sc)}
\UnaryInfC{$\Gamma\left\ob X\DC^{a} Y\bl\context, [P\DC^{b} Q\bl\contextb]\right\cb\left\ob U\DC^{a} V\bl\contexta\right\cb$}
\DisplayProof
\\[1em]
 \AxiomC{$\Gamma\ob A,B,X\DC^{a} Y\bl\context\cb$}
 \RightLabel{${\land}$l}
 \UnaryInfC{$\Gamma\ob A\land B,X\DC^{a} Y\bl\context\cb$}
 \DisplayProof
 &
 \AxiomC{$\Gamma\ob X\DC^{a} Y,A\bl\context\cb$}
 \AxiomC{$\Gamma\ob X\DC^{a} Y,B\bl\context\cb$}
 \RightLabel{${\land}$r}
 \BinaryInfC{$\Gamma\ob X\DC^{a} Y,A\land B\bl\context\cb$}
 \DisplayProof
 \\[1em]
 \AxiomC{$\Gamma\ob A,X\DC^{a} Y\bl\context\cb$}
 \AxiomC{$\Gamma\ob B,X\DC^{a} Y\bl\context\cb$}
 \RightLabel{(${\lor}$l)}
 \BinaryInfC{$\Gamma\ob A\lor B,X\DC^{a} Y\bl\context\cb$}
 \DisplayProof
 &
 \AxiomC{$\Gamma\ob X\DC^{a} Y,A,B\bl\context\cb$}
 \RightLabel{(${\lor}$r)}
 \UnaryInfC{$\Gamma\ob X\DC^{a} Y,A\lor B\bl\context\cb$}
 \DisplayProof
\\[1em]
\AxiomC{$\Gamma\ob A,X\DC^{a} Y\bl\context \cb \ob A,U\DC^{a} V\bl\contexta\cb$}
\RightLabel{(fc-l)}
\UnaryInfC{$\Gamma\ob A,X\DC^{a} Y\bl\context \cb \ob U\DC^{a} V\bl\contexta\cb$}
\DisplayProof
&
\AxiomC{$\Gamma\ob X\DC^{a} Y,A\bl\context\cb\ob U\DC^{a} V,A\bl\contexta\cb$}
\RightLabel{(fc-r)}
\UnaryInfC{$\Gamma\ob X\DC^{a} Y,A\bl\context\cb\ob U\DC^{a} V\bl\contexta\cb$}
\DisplayProof
\end{tabular}
\end{small}
\end{center}

\caption{The INS calculus~$\INSI$.}
\label{fig_INSI}
\end{figure}

\section{Conclusion and Future work}

While extensions of the sequent data structure permit calculi with the subformula property for more logics, the tradeoff between the subformula property in an extended proof formalism and the utility of such a calculus (for proving metatheoretic results) is still unclear. This provides a motivation for investigating the relationship~\cite{Fit12,GorRam12AIML,Ram15JLC} between various proof formalisms to understand their relative expressive power. The work here can be seen as part of this program, more specifically as a classification of various proof formalisms as subsystems of the labelled sequent formalism.

From a \textit{procedural} perspective, (I)NS were obtained by generalising the traditional sequent. Meanwhile LTS(E) can be seen as specific cases of the labelled sequents. Extending a formalism by generalisation has the advantage of intuition: extend just enough to capture the logic of interest without losing nice syntactic properties. Obtaining a formalism by specialisation opens the possibility of coercing existing results to the new situation as we have done here.

This work (Cor.~\ref{cor-INSK}) and~\cite{GorRam12AIML} indicates that the (I)NS calculus cannot be considered `more syntactic' (and hence preferable from a syntactic/proof-theoretic viewpoint) than a labelled sequent calculus. After all, given the tight correspondence between derivations in (I)NS calculi and LTS(E)-derivations in LS-calculi that we have demonstrated here, these two objects must be equally syntactic. Of course, one may also have a preference of formalism based on other grounds e.g. a preference in notation.

By going beyond the restriction of an LTSE-derivation, it may well be the case that further logics are presentable, and some of these logics might not have any INS calculus with subformula property. If this were the case, then it would indicate that the the general LS derivations are more expressive than LTSE-derivations (equivalently, INS). This is an interesting question and the topic of future work. A similar question was raised recently~\cite{CiaLyoRam18} (replacing LTSE-derivations with the image of derivations from the display calculus) and it was observed that the transformations from the general derivation to the restricted form appear non-trivial (if they exist at all).

%
%

Another topic for future work is the application of these proof calculi to prove results about the logic. A very natural question is that of proving decidability. For example, there is a prominent set of Geach axioms for which decidability is open~\cite{WolZak07,ChaZak97}, namely extensions of~$K$ by axioms from $\{ \Box^{i} p\imp \Box^{j} p | 2\leq i<j \}$. The challenge can be phrased as identifying \textit{when} to terminate a backward proof search for a derivation, and building a countermodel from the terminated proof search.
It would be interesting to investigate if the distinctive features of restricted (LTSE) and unrestricted labelled sequent derivations could be exploited in order to handle these two distinct tasks.






 


\section{Acknowledgements}

The author wishes to thank the anonymous referees for their careful reading and helpful comments on a previous version of this paper.

\bibliographystyle{splncs03}
\bibliography{mybib}

\end{document}